\documentclass[11pt,a4paper,english]{article}

\usepackage{amsmath,amsfonts,amssymb}
\usepackage[english]{babel}
\usepackage{amsmath}
\usepackage{nicefrac}
\usepackage{amstext}
\usepackage{amsthm,mathrsfs,mathrsfs}
\usepackage{amsmath,amsfonts,colonequals}
\usepackage{framed,footmisc}
\usepackage{natbib}

\usepackage{url}
\usepackage{color}

\theoremstyle{plain} \newtheorem{Def}{Definition}\numberwithin{Def}{section}
\theoremstyle{plain} \newtheorem{Lemma}[Def]{Lemma}\numberwithin{Def}{section}
\theoremstyle{plain} \newtheorem{prop}[Def]{Proposition}\numberwithin{Def}{section}
\theoremstyle{plain} \newtheorem{Theorem}[Def]{Theorem}\numberwithin{Def}{section}
\theoremstyle{plain} \newtheorem{Remark}[Def]{Remark}\numberwithin{Def}{section}
\theoremstyle{plain} \newtheorem{Corollary}[Def]{Corollary}\numberwithin{Def}{section}
\theoremstyle{plain} \newtheorem{example}[Def]{Example}\numberwithin{Def}{section}
\theoremstyle{plain} 
\numberwithin{equation}{section}

\title{The Long-Term Swap Rate and a General Analysis of Long-Term Interest Rates}
\author{\small Francesca Biagini\thanks{Department of Mathematics, LMU University, Theresienstrasse 39, D-80333 Munich, Germany, email: francesca.biagini@math.lmu.de.}\thanks{Secondary affiliation: Department of Mathematics, University of Oslo, Box 1053, Blindern, 0316, Oslo, Norway.}
, Alessandro Gnoatto\thanks{Department of Economics, University of Verona, Via Cantarane 24, 37129 Verona, Italy email: alessandro.gnoatto@univr.it}, Maximilian H\"{a}rtel\footnotemark[1]}

\begin{document}

\maketitle

\begin{abstract}
We introduce here for the first time the long-term swap rate, characterised as the fair rate 
of an overnight indexed swap with infinitely many exchanges.
Furthermore we analyse the relationship between the long-term swap rate, the long-term yield, 
see \citet{BiaginiGnoattoHaertel}, \citet{bh2012}, and \citet{Karoui97}, 
and the long-term simple rate, considered in \citet{Brody_Hughston} as long-term discounting rate.
We finally investigate the existence of these long-term rates  
in two term structure methodologies, the Flesaker-Hughston model and the linear-rational model. A numerical example illustrates how our results can be used to estimate  the non-optional component of a CoCo bond.
\end{abstract}

\normalsize

\section{Introduction}

The modelling of long-term interest rates is a very important topic for financial institutions investing in securities 
with maturities that have a long time-horizon, such as life insurances or infrastructure projects. 
Most articles focusing on long-term interest rate modelling examine the long-term yield, 
defined as the continuously compounded spot rate where the maturity goes to infinity, as discounting rate for these products, 
cf.\,\,\citet{BiaginiGnoattoHaertel}, \citet{bh2012}, \citet{dybvig96}, \citet{Karoui97}, or \citet{Yao2000}. 
An important result which characterises the long-term yield is the Dybvig-Ingersoll-Ross (DIR) theorem, which states that the long-term 
yield is a non-decreasing process. It was first shown in \citet{dybvig96} and then discussed in \citet{article_Goldammer}, 
\citet{article_Hubalek}, \citet{article_Kardaras}, \citet{article_McCulloch}, and \citet{article_Schulze}. 
According to \citet{Brody_Hughston} the DIR theorem ultimately implies that discounted cashflows with higher time-to-maturity are 
over-penalised, so that the use of this long-term interest rate becomes unsuitable for the valuation of projects 
having maturity in a distant future. To overcome this issue, in \citet{Brody_Hughston} the authors propose to use for discounting the long-term 
simple rate, which is defined as the simple spot rate with an infinite maturity. 
Motivated by this ongoing discussion in the literature, we investigate in this paper alternative long-term interest rates.
\par
We introduce here for the first time the long-term swap rate, which we define as the fair fixed rate of a fixed to floating swap with infinitely many exchanges. 
To the best of our knowledge, there has not been any attempt in the literature to study the long-term swap rate so far.
In particular, we focus our attention on swap rates because, unlike $0$-coupon bonds, they are directly  observable on the market.
Our interest in the long term swap rate is also motivated by the observation that some financial products may involve the interchange of cashflows on a possibly unlimited time horizon. 
This is the case of some kind of contingent convertible (CoCo) bonds, which became popular after the financial crisis in 2008. 
Such products are debt instruments issued by credit institutes, which embed the option for the bank to convert debt into equity, typically in order to overcome the situation where the bank is not capitalised enough (cf.\,\,\citet{Coco_1}, \citet{article_BrigoGarciaPede}, \citet{Coco_5}, \citet{Coco_4}, and \citet{Coco_2}).
In the course of the crisis the importance of CoCo bonds for financial institutions to maintain a certain level of capital was pointed out in \citet{Coco_6}. 
In \citet{Coco_7}, the increase in their use  in systemically relevant financial institutions 
was one of three main points that should be realised in the aftermath of the crisis to strengthen the financial system. 
\par
As reported in \citet{article_BrigoGarciaPede}, the value of these instruments may be decomposed as a portfolio consisting of plain bonds and exotic options.
A valuation method for CoCo bonds with finite maturity is presented in \citet{article_BrigoGarciaPede}, 
whereas \citet{Coco_1} also considers the case of unlimited maturity.
Such a result is of practical importance since some of these products offered on the market have maturity equal to infinity (cf.\,\,\citet{Coco_Termsheet}).
In a situation where the CoCo bond has infinite maturity and the coupons of the non-optional part are floating,
it is then natural to ask for an instrument which allows to hedge the interest rate risk involved in the non-optional part of the contract. 
A fixed to floating interest rate swap with infinitely many exchanges could serve as a hedging product for the interest rate risk beared by CoCo bonds. 
The main input for defining such a swap is its fixed rate, i.e.\,\,the long-term swap rate.
Furthermore the long-term swap rate may also play an important role in the context of multiple curve bootstrapping. 
As we shall see in the following, we will concentrate our investigations on overnight indexed swap (OIS) contracts. 
Such OIS contracts constitute the input quotes for bootstrapping procedures which allow for the construction of a discounting curve, 
according to the post-crisis market practice (cf.\,\,for example \citet{cuchieroFontanaGnoatto} or \citet{Book_Henrard}). 
In view of this, the long-term swap rate becomes quite a natural object, from which information on the long-end of the discounting 
curve can be inferred.
\par
The main result of the paper is then the definition of the long-term swap rate $R$ and the study of its properties and relations with the long-term yield and the long-term simple rate.
In particular, we obtain that the long-term swap rate always exists finitely and that this rate is either constant or non-monotonic. 
In the case of a convergent infinite weighted sum $S_{\infty}$ of bonds, we are able to provide an explicit model-independent formula for $R$, 
which is only dependent on $S_{\infty}$, see \eqref{equ:Prop_5_1}.
Hence the long-term swap rate could represent an alternative discounting tool for long-term investments, 
since it is always finite, non-monotonic, can be explicitly characterised, and can be inferred by products existing on the markets. 
\par
As a contribution to the ongoing discussion on suitable discounting factors for investments over long time horizons, 
we then provide a comprehensive analysis of the relations among the long-term yield, the long-term simple rate, and the long-term swap rate in a model-free approach. 
In particular, we study how the existence of one of these long-term rates impacts the existence and finiteness of the other ones. 
This analysis shows the advantage of using the long-term swap rate as discounting rate, since it always remains finite when the other rates may become zero or explode.
\par
The paper is structured as follows. 
First, we introduce in Section \ref{Fixed_Income_Section} some necessary prerequisites, 
such as the different kinds of interest rates and interest rate swaps, in particular OISs. 
Then, Sections \ref{Long_Term_Rates_Section} and \ref{Long_Swap_Rate_Section} describe the three asymptotic rates and some important features of the long-term swap rate like the model-free formula. 
In Section \ref{Relation_between_Long_Term_Rates} we investigate the influence of each long-term rate on the existence and 
finiteness of the other rates.
Finally, in Section \ref{Specific_Term_Structure} we analyse the long-term rates in some selected term structure models. 
We chose the Flesaker-Hughston methodology, developed in \citet{Flesaker_Hughston1996}, 
and the linear-rational term structure methodology, presented in \citet{Filipovic2014},
since they also include the wide class of affine interest models and possess some appealing features such as high tractability and simple forms of the different interest rates. 
In both cases we compute the long-term swap rate and the other long-term rates. We conclude with a numerical example, where we illustrate how our results can be used to estimate  the non-optional component of a CoCo bond.

\section{Fixed Income Setup}\label{Fixed_Income_Section}

\subsection{Interest Rates}\label{Rates}

We now introduce some notations. All quantities in the following are assumed to be associated to a risk-free curve, which, 
in the post-crisis market setting, can be approximated by the overnight curve used in collateralised transactions (cf.\,\,Section 1.1 of \citet{cuchieroFontanaGnoatto}). 
\par
First, we define the contract value of a zero-coupon bond at time $t$ with maturity $T > t$ as $P\!\left(t,T\right)$. 
It guarantees its holder the payment of one unit of currency at time $T$, hence $P\!\left(T,T\right) = 1$ for all $T \geq 0$. 
We assume that there exists a frictionless market for zero-coupon bonds for every time $T > 0$ and that $P\!\left(t,T\right)$ is differentiable in $T$. 
In the following we consider a probability space $\left(\Omega, \mathcal{F}, \mathbb{P}\right)$ endowed with 
the filtration $\mathbb{F} \colonequals \left(\mathcal{F}_{t}\right)_{t \geq 0}$ satisfying the usual hypothesis 
of right-continuity and completeness. 
Furthermore, we only consider finite positive zero-coupon bond prices, i.e.\,\,$0 < P\!\left(t,T\right) < +\infty$ $\mathbb{P}$-a.s.\,\,for all $0 \leq t \leq T$.  
Then, we define the yield for $\left[t,T\right]$ as the continuously compounded spot rate for $\left[t,T\right]$
\begin{equation}\label{Yield_def_1}
 Y\!\left(t,T\right) \colonequals - \frac{\log P\!\left(t,T\right)}{T-t}.
\end{equation}
The simple spot rate for $\left[t,T\right]$ is 
\begin{equation}\label{Simple_def}
 L\!\left(t,T\right) \colonequals  \frac{1}{T - t} \left(\frac{1}{P\!\left(t,T\right)} - 1\right).
\end{equation}
The short rate at time $t$ is defined as 
\begin{equation}\label{Short_def}
 r_{t} \colonequals \lim_{T \downarrow t} Y\!\left(t,T\right)\,\,\,\mathbb{P}\text{-a.s.}
\end{equation}
The corresponding money-market account is denoted by $\left(\beta_{t}\right)_{t \geq 0}$ with 
\begin{equation}\label{MoneyMarket_def}
 \beta_{t} \colonequals \exp\!\left(\int_{0}^{t} r_{s}\, ds\right).
\end{equation}
\par
In particular we assume an arbitrage-free setting, where the discounted bond price process $\frac{P\left(t,T\right)}{\beta_{t}}, t \in \left[0,T\right]$, 
is an $\left(\mathbb{F}, \mathbb{P}\right)$-martingale for all $T > 0$. This implies that the large financial market, consisting of infinitely many bonds, is arbitrage free in the sense of \textit{no asymptotic free lunch with vanishing risk}, see \cite{Cuchiero2016} Assumption 2.2 and \cite{Cuchiero2018}. This also implies that $\beta$ is well-defined, i.e.\,\,$\int_{0}^{t} |r_{s}|\, ds < \infty$ a.s.\,\,for all $t \geq 0$.
We assume to work with the c\`{a}dl\`{a}g version of $\frac{P\left(t,T\right)}{\beta_{t}}, t \in \left[0,T\right]$, for all $T > 0$. 
Consequently $P\!\left(t,T\right), Y\!\left(t,T\right), L\!\left(t,T\right), t \in \left[0,T\right]$, are all c\`{a}dl\`{a}g processes in the sequel. 


\subsection{Interest Rate Swaps}\label{IRS}

Swap contracts are derivatives where two counterparties exchange cashflows. 
There exist different kinds of swap contracts, involving cashflows deriving for example from 
commodities, credit risk or loans in different currencies. As far as interest rate swaps are concerned, 
the evaluation of such claims represents an aspect which is part of the discussion on multiple curve models, 
due to the recent financial crisis. While a survey of the literature on multiple-curve models would be beyond the scope of 
the present paper\footnote{For a complete list of references the interested reader is referred to \citet{cuchieroFontanaGnoatto}.}, 
we limit ourselves to note that even in the post-crisis setting, there are particular types of interest rate swaps whose 
evaluation formulas are equivalent to the ones employed for standard interest rate swaps in the single-curve pre-crisis setting. 
Since such instruments, called OISs, play a pivotal role in the construction of discount curves, 
we concentrate our study on them, and avoid to define a full multiple-curve model.
\par
We consider a infinite tenor structure of the form
\begin{equation}\label{Tenor_2}
0< T_{0} < T_{1} < \dots < T_{n}  <\cdots,
\end{equation}
for $n \in \mathbb{N}$.
We set $\delta_{i} \colonequals T_{i} - T_{i-1}, i\in  \mathbb{N}\! \setminus\! \left\{0\right\}$.
In an OIS contract, floating payments are indexed to a compounded overnight rate like EONIA. 
The variable rate that one party has to pay every time $T_{i}$, $i = 1, 2, \dots$, is $\delta_i \bar{L}\!\left(T_{i-1},T_{i}\right)$ 
with $\bar{L}\!\left(T_{i-1},T_{i}\right)$ denoting the compounded overnight rate for $\left[T_{i-1},T_{i}\right]$. 
This rate is given by (cf.\,\,equation (10) of \citet{FilipovicTrolle}) 
\begin{equation*}
 \bar{L}\!\left(T_{i-1},T_{i}\right) = \frac{1}{\delta_i} \left(\exp\!\left(\int_{T_{i-1}}^{T_{i}} r_{s} \, ds\right) - 1\right)\,.
\end{equation*}
Fixed $n$, the OIS rate for the period $[T_0,T_n]$, i.e.\,\,the fixed rate which makes the OIS value equal to zero at inception, is for $t\leq T_0$
\begin{align}\label{OIS_Rate}
 R^{OIS}\!\left(t;T_{0},T_n\right) & = \frac{\sum_{i=1}^{n} \mathbb{E}^{\mathbb{P}}\!\left[\left.\exp\!\left(-\int_{t}^{T_{i}} r_{s}\,ds\right)\delta_i \bar{L}\!\left(T_{i-1},T_{i}\right) \right| \mathcal{F}_{t}\right]}{\sum_{i=1}^{n} P\!\left(t,T_{i}\right)\delta_i}  \nonumber\\
& = \frac{P\!\left(t,T_{0}\right) - P\!\left(t,T_n\right)}{\sum_{i=1}^{n} P\!\left(t,T_{i}\right)\delta_i}
\end{align}
or in general
\begin{displaymath}
R^{OIS}\!\left(t;T_{0},T_n\right) = \frac{P\!\left(t,T_{0}\right) - P\!\left(t,T_n\right)}{\int_{T_{1}}^{T_n}\! \exp\!\left(-\left(s-t\right) Y\!\left(t,s\right)\right) \xi\!\left(ds\right)}\,,
\end{displaymath}
where $\xi$ is a measure on $\left(\mathbb{R}_{+},\mathcal{B}\!\left(\mathbb{R}_{+}\right)\right)$ and $Y$ is the yield, defined in \eqref{Yield_def_1}.
Note that \eqref{OIS_Rate} corresponds to the formula for the par swap rate in a single curve setting. 
In the following, we consider only OIS swaps and set 
\begin{equation}\label{OIS_Rate_1}
R\!\left(t,T_n\right) \colonequals R^{OIS}\!\left(t;t,T_n\right)
\end{equation}
for all $t\leq T_n$. 

\begin{Remark}
Note that in \eqref{OIS_Rate_1} we have set $T_{0} = t$. 
This is equivalent to consider the interest rate $R\!\left(t,T\right)$ as associated to a rolling over of OIS contracts.
This is possible in our model since we admit the existence of bonds for any maturity $T > 0$.
\end{Remark}

\section{Long-Term Rates}\label{Long_Term_Rates_Section}


In this section we consider some possible long-term rates. 
In particular we focus on the long-term yield and the long-term simple rate, which have been already defined in the literature (cf.\,\,\citet{Karoui97} and \citet{Brody_Hughston}). 
The \emph{long-term yield} can be defined in different ways. Some articles investigate interest rates with a certain time to maturity to approach the 
concept of ``long-term'', e.g.\,\,in \citet{Yao2000} yield curves with time to maturity over 30 years are examined, \citet{Shiller79} considers
yields with a maturity beyond 20 years to be ``long-term'', whereas the ECB takes 10 years as a barrier, cf.\,\,\citet{ECB2015}. 
Another approach is to look at the asymptotic behaviour of the yield curve by letting the maturity go to infinity. 
This approach is used by \citet{BiaginiGnoattoHaertel}, \citet{bh2012}, \citet{dybvig96}, \citet{Karoui97}. 
In line with the above-mentioned principle, we introduce our first object of study, and define the long-term yield 
$\ell \colonequals \left(\ell_{t}\right)_{t\geq 0}$ as 
\begin{equation}\label{Long_Yield}
 \ell_{\,\cdot} \colonequals \lim_{T \rightarrow \infty} Y\!\left(\,\cdot\,,T\right),
\end{equation}
if the limit exists in the sense of the uniform convergence on compacts in probability (convergence in ucp).\footnote{For a definition of the ucp convergence and some additional results the reader is referred to Section \ref{appendix_ucp} in the appendix.} 
If the limit in \eqref{Long_Yield} exists but it is infinite, positive or negative, see Definition \ref{UCP_Def_2}, we will write $\ell = \pm \infty$ for the sake of simplicity. 
We will use this improper notation also for the other long-term interest rates in the sequel of the paper.
We recall that the long-term yield process $\ell$ is a non-decreasing process by the DIR theorem (cf.\,\,Theorem 2 of \citet{dybvig96}), 
which was first proved in \citet{dybvig96} and further discussed in \citet{article_Goldammer}, \citet{article_Hubalek}, and \citet{article_Kardaras}. 
\par
In \citet{Brody_Hughston} it is suggested to consider a particular model for the \emph{long-term simple rate} for the discounting of cashflows occuring in a distant future. 
By using exponential discount factors the discounted value of a long-term project, that will be realised over a long time horizon, 
in most cases will turn out to be overdiscounted, hence too small to justify the overall project costs. 
To overcome this problem, the authors of \citet{Brody_Hughston} came up 
with the concept of ``social discounting'', where the long-term simple rate is employed for discounting cashflows in the distant future. 
To integrate this interesting approach into our considerations, we now define the long-term simple rate process $L \colonequals \left(L_{t}\right)_{t\geq 0}$ as 
\begin{equation*}
 L_{\,\cdot} \colonequals \lim_{T \rightarrow \infty} L\!\left(\,\cdot\,,T\right),
\end{equation*}
if the limit exists in ucp, where $L\!\left(t,T\right)$ is defined in \eqref{Simple_def}. 
Note that $L_{t} \geq 0$ $\mathbb{P}$-a.s.\,\,for all $t\geq 0$ by \eqref{Simple_def}.
\par
We introduce the process 
$P \colonequals \left(P_{t}\right)_{t\geq 0}$ as
\begin{equation}\label{Long_Bond}
 P_{\,\cdot} \colonequals \lim_{T \rightarrow \infty} P\!\left(\cdot,T\right),
\end{equation}
if the limit exists (finite or infinite) in ucp.  For an alternative definition of long bond, see  \cite{Qin-Linetsky}.

\begin{Remark}\label{Remark_1}
We note that as a consequence of our assumption that the bond prices are c\`{a}dl\`{a}g, we also obtain that all the 
long-term rates introduced above and in Section \ref{Long_Swap_Rate_Section} are c\`{a}dl\`{a}g. In the sequel we will then use Theorem 2 of Chapter I, Section 1 
of \citet{Book_Protter}, which tells us that for two right-continuous stochastic processes $X$ and $Y$ it holds that 
$X_{t} = Y_{t}$ $\mathbb{P}$-a.s.\,\,for all $t \geq 0$ is equivalent to $\mathbb{P}$-a.s.\,\,for all $t \geq 0$, $X_{t} = Y_{t}$.
\end{Remark}

We define $S_{n} \colonequals \left(S_{n}\!\left(t\right)\right)_{t \geq 0}$ with
\begin{equation*}
 S_{n}\!\left(t\right) \colonequals\int_{T_{1}}^{T_{n}}\! \exp\!\left(-\left(T-t\right)Y\!\left(t,T\right)\right) \xi\!\left(dT\right), t\geq 0,
\end{equation*}
considering a tenor structure with infinite many exchange dates. It is clear that
\begin{equation}\label{Sum_Bond_Prices}
 S_{n}\!\left(t\right)= \sum_{i = 1}^{n} \delta_i P\!\left(t,T_{i}\right), t\geq 0,
\end{equation}
if $\xi\!\left(dT\right)=\sum_{i=1}^{+\infty} \delta_i \delta_{\left\{T_i\right\}} $, with $\delta_{i} \colonequals T_{i} - T_{i-1}$ and Dirac's delta functions $\delta_{\left\{T_i\right\}}$.
Then the limit 
\begin{equation}\label{Sum_Infinite_Bond_Prices}
\lim_{n \rightarrow \infty} S_{n}\!\left(\cdot\right)
\end{equation}
in ucp always exists, finite or infinite. 
In the sequel we denote this limit by $S_{\infty}$ if it exists and is finite.
All bond prices are strictly positive, therefore for all $t \geq 0, n \in \mathbb{N}$ 
we have $\mathbb{P}$-a.s.\,\,$S_{n}\!\left(t\right) > 0$ and $S_{\infty}\!\left(t\right) > 0$.

\section{The Long-Term Swap Rate}\label{Long_Swap_Rate_Section}

We now introduce the \emph{long-term swap rate} $R \colonequals \left(R_{t}\right)_{t\geq 0}$ as
\begin{equation*}
 R_{\,\cdot\,} \colonequals \lim_{n \rightarrow \infty} R\!\left(\,\cdot\,,T_{n}\right) 
\end{equation*}
if the limit exists in ucp, where $R\!\left(t,T_n\right)$ is defined in \eqref{OIS_Rate_1}. 
The long-term swap rate, defined here for the first time, can be understood as the fair fixed rate of an OIS starting in $t$ that has a payment stream 
with infinitely many exchanges. This fixed rate is meant to be fair in the sense that the initial value of this OIS equals zero.
\par
We investigate the existence and finiteness of the long-term swap rate. 
We first provide a model-free formula for the swap rate, when $S_{\infty}$ exists and is finite. 
In particular, we focus here on the case when $S_n$ is given by \eqref{Sum_Bond_Prices} and  the tenor structure is such that
\begin{equation}\label{tenor_nondeg}
c < \inf_{i \in \mathbb{N} \setminus\! \left\{0\right\}} \left(T_{i} - T_{i-1}\right) 
\end{equation}
with $c >0$.  This hypothesis avoids the degenerated case where $|T_i -T_{i-1}|\rightarrow 0$ for $i\rightarrow \infty$, and corresponds to the realistic setting of a fixed tenor (but where the number of dates may become very large).

This section relies on some properties of $S_{\infty}$, which we proved in Appendix A.

\begin{Theorem}\label{Prop_5}
Assume that $S_n$ is defined as in \eqref{Sum_Bond_Prices} for $n\in \mathbb{N}$ and the tenor structure satisfy condition \eqref{tenor_nondeg}.
\begin{itemize}
\item[(i)] 
  If $S_{n} \overset{n \rightarrow \infty}{\longrightarrow} S_{\infty}$ in ucp, then $\mathbb{P}$-a.s.
  \begin{equation}\label{equ:Prop_5_1}
   R_{t} = \frac{1}{S_{\infty}\!\left(t\right)} > 0
  \end{equation}
  for all $t \geq 0$.
\item[(ii)] 
  If $S_{n} \overset{n \rightarrow \infty}{\longrightarrow} +\infty$ in ucp and
  $P$, defined in \eqref{Long_Bond}, exists finitely, then 
  it holds $R_{t} = 0$ $\mathbb{P}$-a.s.\,\,for all $t\geq 0$.
\end{itemize}
\item[(iii)]
  The long-term swap rate cannot explode, i.e.\,\,$\mathbb{P}\!\left(\left\vert R_{t}\right\vert < +\infty\right) =1$    for all $t \geq 0$.
\end{Theorem}

\begin{proof}
To (i): We have that in ucp
\begin{align*}
 \lim_{n \rightarrow \infty} R\!\left(\,\cdot\,,T_{n}\right) & \overset{\eqref{OIS_Rate_1}}{=} 
 \lim_{n \rightarrow \infty} \frac{1 - P\!\left(\,\cdot\,,T_{n}\right)}{ S_{n}\!\left(\,\cdot\,\right)} 
= \lim_{n \rightarrow \infty} \frac{1}{S_{n}\!\left(\,\cdot\,\right)} - \lim_{n \rightarrow \infty} \frac{P\!\left(\,\cdot\,,T_{n}\right)}{S_{n}\!\left(\,\cdot\,\right)} 
\nonumber\\
 & \overset{\phantom{\eqref{OIS_Rate_1}}}{=} \lim_{n \rightarrow \infty} \frac{1}{S_{n}\!\left(\cdot\right)} 
 = \frac{1}{S_{\infty}\!\left(\,\cdot\,\right)} > 0
\end{align*}
by Theorem \ref{UCP_Thm_1}.
\newline
To (ii): We have that in ucp
\begin{align*}
  \lim_{n \rightarrow \infty} R\!\left(\,\cdot\,,T_{n}\right) & \overset{\eqref{OIS_Rate_1}}{=}
 \lim_{n \rightarrow \infty} \frac{1 - P\!\left(\,\cdot\,,T_{n}\right)}{S_{n}\!\left(\,\cdot\,\right)} \nonumber\\
 & \overset{\phantom{\eqref{OIS_Rate_1}}}{=}  \lim_{n \rightarrow \infty} \frac{1}{S_{n}\!\left(\,\cdot\,\right)} - \lim_{n \rightarrow \infty} \frac{P\!\left(\,\cdot\,,T_{n}\right)}{S_{n}\!\left(\,\cdot\,\right)} \nonumber\\
 & \overset{\phantom{\eqref{OIS_Rate_1}}}{=}   \lim_{n \rightarrow \infty} \frac{1}{S_{n}\!\left(\cdot\right)} - \lim_{n \rightarrow \infty} \frac{P\!\left(\,\cdot\,,T_{n}\right)}{S_{n}\!\left(\,\cdot\,\right)} \nonumber\\
 & \overset{\phantom{\eqref{OIS_Rate_1}}}{=}   - \lim_{n \rightarrow \infty} \frac{P\!\left(\,\cdot\,,T_{n}\right)}{S_{n}\!\left(\,\cdot\,\right)}
 = - P_{\,\cdot\,}\lim_{n \rightarrow \infty} \frac{1}{S_{n}\!\left(\cdot\right)} = 0
\end{align*}
by Theorem \ref{UCP_Thm_1}.
\newline
To (iii): Since (i) and (ii) hold, we need only to study the case when $S_{n} \overset{n \rightarrow \infty}{\longrightarrow} +\infty$ in ucp and
  $P=+\infty$. We have that in this case\begin{equation*}
  \lim_{n \rightarrow \infty} R\!\left(\,\cdot\,,T_{n}\right)  =- \lim_{n \rightarrow \infty} \frac{P\!\left(\,\cdot\,,T_{n}\right)}{S_{n}\!\left(\,\cdot\,\right)}
\end{equation*}
   in ucp. We  note that $\mathbb{P}$-a.s.
\begin{equation*}
 0 \leq \sup_{0\leq s\leq t} \frac{P\!\left(s,T_{n}\right)}{S_{n}\!\left(s\right)} = \sup_{0\leq s\leq t}\frac{1}{\delta_n} \left(1 - \frac{S_{n-1}\!\left(s\right)}{S_{n}\!\left(s\right)}\right) \leq 1/c
\end{equation*}
for all $t\geq 0$ with $\delta_n=T_n-T_{n-1}$. 
Hence 
\begin{equation*}
 \mathbb{P}\!\left(\inf_{0\leq s\leq t}\frac{P\left(s,T_{n}\right)}{S_{n}\!\left(s\right)} > M\right) \overset{n \rightarrow \infty}{\longrightarrow} 0
\end{equation*}
for all $M > 1/c$. This contradicts Definition \ref{UCP_Def_2} of ucp convergence to $+\infty$ applied to $|R|$, so the long-term swap rate always exists and is finite $\mathbb{P}$-a.s..
\end{proof}

\begin{Remark}\label{consol}
If now we consider a tenor structure with $T_{i} - T_{i-1} = \delta$ for all $i \in \mathbb{N}\! \setminus\! \left\{0\right\}$,
then $S_{n}\!\left(t\right) = \delta \sum_{i = 1}^{n}\! P\!\left(t,T_{i}\right)$ and \eqref{equ:Prop_5_1} boils down to
\begin{displaymath}
R_{t} = \frac{1}{\delta \sum_{i = 1}^{\infty}\! P\!\left(t,T_{i}\right)}\,,
\end{displaymath}
i.e.\,\,$R_{t}$ is proportional to the consol rate of a perpetual bond.
For more details on consol bonds, we refer to \citet{Delsaen}, \citet{DuffieJinMaYong}, and the references therein. 
However, our construction is more general and goes beyond the existence of the consol bond rate.
Our approach has the advantage of being consistent with the multi-curve theory of interest rate modelling as well as of delivering  a long-term interest rate which is always finite.
\end{Remark}

By Theorem \ref{Prop_5} (i) and (ii) we obtain the existence of the long-term swap rate as a finite limit 
if $P$ exists finitely. However this result always holds as shown by Theorem \ref{Prop_5} (iii).

\begin{Corollary}\label{Cor_7}
If $S_{n} \overset{n \rightarrow \infty}{\longrightarrow} +\infty$ in ucp, then 
it holds $-\infty < R_{t} \leq 0$ $\mathbb{P}$-a.s.\,\,for all $t\geq 0$.
\end{Corollary}
\begin{proof} 
This is a consequence of Theorem \ref{Prop_5} (ii) and (iii).
\end{proof}

\begin{prop}\label{Prop_Short_Rate}
 Suppose $S_{n} \overset{n \rightarrow \infty}{\longrightarrow} +\infty$ in ucp. If for all $n \in \mathbb{N}$ $P\!\left(t,T_{n}\right) \geq P\!\left(t,T_{n+1}\right)$ for all $t \in \left[0,T_{n}\right]$, then
 \begin{equation*}
  R_{t} = - k_{t}
  \end{equation*}
for a process $\left(k_{t}\right)_{t\geq 0}$ with $0\leq k_{t} \leq 1$ $\mathbb{P}$-a.s.\,\,for all $t\geq 0$.
\end{prop}
\begin{proof}
Since for all $n \in \mathbb{N}$, $S_{n}\!\left(t\right) \leq S_{n+1}\!\left(t\right)$ $\mathbb{P}$-a.s.\,\,for all $t\geq 0$, 
for all $n \in \mathbb{N}$ we have $\mathbb{P}$-a.s.
\begin{equation*}
 \frac{P\!\left(t,T_{n}\right)}{S_{n}\!\left(t\right)} 
 = \frac{P\left(t,T_{n}\right)}{\beta_{t}} \frac{\beta_{t}}{S_{n}\!\left(t\right)}
 \geq \frac{P\left(t,T_{n+1}\right)}{\beta_{t}} \frac{\beta_{t}}{S_{n + 1}\!\left(t\right)}
 = \frac{P\!\left(t,T_{n+1}\right)}{S_{n+1}\!\left(t\right)}
\end{equation*}
for all $t\geq 0$. This implies that $\mathbb{P}$-a.s.
\begin{equation*}
 1 \geq \sup_{0\leq s\leq t} \frac{P\!\left(s,T_{n}\right)}{S_{n}\!\left(s\right)} \geq \sup_{0\leq s\leq t} \frac{P\!\left(s,T_{n+1}\right)}{S_{n+1}\!\left(s\right)}
\end{equation*}
for all $t\geq 0$. Hence $\frac{P\left(\,\cdot\,,T_{n}\right)}{S_{n}\left(\cdot\right)} \overset{n \rightarrow \infty}{\longrightarrow} k_{\cdot}$ in ucp, with $0\leq k_{t} \leq 1$ $\mathbb{P}$-a.s.\,\,for all $t\geq 0$. 
In particular  by Theorem \ref{Prop_5} (ii), we get $k_{t} = 0$ $\mathbb{P}$-a.s.\,\,for all $t\geq 0$ if $P_{t} < +\infty$ $\mathbb{P}$-a.s.\,\,for all $t\geq 0$.
\end{proof}

\begin{Remark}
\begin{enumerate}
\item Note that if $r_{t} \geq 0$ for all $t\geq 0$, then bond prices are decreasing with respect to time to maturity.
\item We remark that the process $k$ in Proposition \ref{Prop_Short_Rate} may not necessarily be identically zero if $S_{n} \overset{n \rightarrow \infty}{\longrightarrow} +\infty$ in ucp. In particular consider the (unrealistic) case when $P(t,T_n):=1 + (T-t)^n$ for some $T>0$.
Then 
$$\frac{P\left(t,T_{n}\right)}{S_{n}(t)}= \frac{1 + (T-t)^n}{\sum_{i=1}^n (1 + (T-t)^i)}=  \frac{1}{1 +\sum_{i=1}^{n-1} \frac{1 + (T-t)^i}{1 + (T-t)^n}}\overset{n \rightarrow \infty}{\longrightarrow} 1.$$ 
\end{enumerate}
\end{Remark}

If we assume that there exists a liquid market for perpetual OIS, 
meaning OIS with infinitely many exchanges with the fixed rate corresponding to the long-term swap rate, 
we can state the following theorem. We recall that we are working under the hypothesis that $\mathbb{P}$ 
is an equivalent martingale measure for the bond market, i.e\,\,that the bond market is arbitrage-free in 
the sense of \textit{no asymptotic free lunch with vanishing risk}, see \cite{Cuchiero2016}. 

\begin{Theorem}\label{Theorem_Non_Monotonic}
 In the setting outlined in Section \ref{Rates}, the long-term swap rate is either constant or non-monotonic.
\end{Theorem}

\begin{proof}
 First, we assume that $R_{s} \geq R_{t}$ $\mathbb{P}$-a.s.\,\,with $\mathbb{P}\!\left(R_{s} > R_{t}\right) > 0$ for $0 \leq t < s$. 
 Then, let us consider the following investment strategy. 
 At time $t$ we enter a payer OIS with perpetual annuity, nominal value $N$, fixed-rate $R_{t}$ and the following tenor structure
 \begin{equation}\label{Theorem_Equ_0}
  t < s \leq T_{1} < \cdots < T_{n}\,
 \end{equation}
 where $n \rightarrow \infty$.
 This investment has zero value in $t$, so there is no net investment so far. 
 We receive the following payoff in each $T_{i}$, $i \in  \mathbb{N}\! \setminus\! \left\{0\right\}$:
 \begin{equation*}
  \left(\bar{L}\!\left(T_{i-1},T_{i}\right) - R_{t}\right) \delta_i N\,.
 \end{equation*}
Then at time $s$ we enter a receiver OIS with a perpetual annuity, nominal value $N$, a fixed-rate of $R_{s}$ and the same tenor structure as in \eqref{Theorem_Equ_0}. 
The value of this OIS is zero in $s$, hence there is still no net investment, and the payoff in each $T_{i}$, $i \in \mathbb{N}$, resulting from this OIS is: 
\begin{equation*}
  \left(R_{s} - \bar{L}\!\left(T_{i-1},T_{i}\right)\right) \delta_i N\,.
 \end{equation*}
 This strategy leads to the payoff at $T_i$
 \begin{equation*}
  H_{i} \colonequals \left(\bar{L}\!\left(T_{i-1},T_{i}\right) - R_{t}\right) \delta_i N + \left(R_{s} - \bar{L}\!\left(T_{i-1},T_{i}\right)\right) \delta_i N 
  = \delta_i N \left(R_{s} - R_{t}\right) \geq 0
 \end{equation*}
 with $\mathbb{P}\!\left(H_{i} > 0\right) > 0$, i.e.\,\,to an arbitrage.
 \par
 If we assume that $R_{s} \leq R_{t}$ $\mathbb{P}$-a.s.\,\,with $\mathbb{P}\!\left(R_{s} < R_{t}\right) > 0$ for $0 \leq t \leq s \leq T_{1}$, 
 we use an analogue arbitrage strategy with the only difference that we invest in $t$ in a receiver OIS and in $s$ in a payer OIS.
 \par
 It follows that in an arbitrage-free market setting the long-term swap rate cannot be non-decreasing or non-increasing, 
 i.e.\,\,it can only be monotonic if it is constant.
\end{proof}


\section{Relation between Long-Term Rates}\label{Relation_between_Long_Term_Rates}

We now study the relation among the long-term rates introduced in Sections \ref{Long_Term_Rates_Section} and \ref{Long_Swap_Rate_Section} in terms of their existence. 
For further details, we also refer to \citet{Thesis}. 
\par
For the sake of simplicity we now assume a tenor structure with $T_{i} - T_{i-1} = \delta$ for all $i \in \mathbb{N}\! \setminus\! \left\{0\right\}$.
This is of course the case when we extrapolate the long-term swap rate by OIS contracts existing on the market.
We choose this setting in order to focus on the behaviour of the bond price for $T \rightarrow \infty$, i.e.\,\,of $P$ defined in \eqref{Long_Bond},
independently of the maturity distances in the tenor structure.

\subsection{Influence of the Long-Term Yield on Long-Term Rates}\label{Influence_Yield_Swap_Simple}

In this section we study the influence of the existence of the long-term yield on the existence of the long-term swap and simple rates. 
Since typical market data indicate positive long-term yields\footnote{For long-term interest rate market data  
please refer to \citet{ECB2015} for the EUR market and to \citet{FED2015} for the USD market.}, we restrict ourselves to the cases of $\ell \geq 0$. 
For a more general analysis which also takes into account the possibility of a negative long-term yield, we refer to \citet{Thesis}.


\begin{Theorem}\label{Prop_8}
 If $0 < \ell_{t} < +\infty$ $\mathbb{P}$-a.s.\,\,for all $t\geq 0$, then $0 < R_{t} < +\infty$ $\mathbb{P}$-a.s.\,\,for all $t\geq 0$ and $L = + \infty$.
\end{Theorem}

\begin{proof}
 First, we show that $S_{n} \overset{n \rightarrow \infty}{\longrightarrow} S_{\infty}$ in ucp. 
 For this, it is sufficient to show that for all $t\geq 0$, $\lim_{n \rightarrow \infty} \sup_{0\leq s\leq t} S_{n}\!\left(s\right) < +\infty$ 
$\mathbb{P}$-a.s., since this also implies $\lim_{n \rightarrow \infty} \sup_{0\leq s\leq t} S_{n}\!\left(s\right) < +\infty$ in probability.
\par
We know that for all $t\geq 0$ and all $\epsilon > 0$ it holds 
\small
\begin{align*}
 \mathbb{P}\!\left(\sup_{0\leq s\leq t} \left\vert Y\!\left(s,T_{n}\right) - \ell_{s} \right\vert \leq \epsilon\right) 
 & \overset{\eqref{Yield_def_1}}{=} \mathbb{P}\!\left(\sup_{0\leq s\leq t} \left\vert \frac{\log\!P\!\left(s,T_{n}\right)}{T_{n}-s} + \ell_{s} \right\vert \leq \epsilon \right)
  \overset{n \rightarrow \infty}{\longrightarrow} 1,
\end{align*}
\normalsize
i.e.\,\,for all $t\geq 0$ and all $\epsilon > 0$ there exists $N_{\epsilon}^{t} \in \mathbb{N}$ such that for all $n \geq N_{\epsilon}^{t}$ 
\begin{equation}\label{inequ:Prop_8_1}
  \mathbb{P}\!\left(\sup_{0\leq s\leq t} \left\vert \frac{\log\!P\!\left(s,T_{n}\right)}{T_{n}-s} + \ell_{s} \right\vert \leq \epsilon \right) > 1 - \delta\!\left(\epsilon\right)
\end{equation}
with $\delta\!\left(\epsilon\right) \rightarrow 0$ for $\epsilon \rightarrow 0$. 
Define for $\epsilon > 0$, $u\geq 0$ and $n \in \mathbb{N}$
\begin{equation}\label{def:A_2}
 A^{\epsilon, u, n}_{1} \colonequals  \left\{\omega \in \Omega: \sup_{0\leq s\leq u} \left\vert \frac{\log\!P\!\left(s,T_{n}\right)}{T_{n}-s} +  \ell_{s}\right\vert \leq \epsilon \right\}.
\end{equation}
Then for $n \geq N_{\epsilon}^{u}$ with $u > t$ we have $\mathbb{P}\!\left(A^{\epsilon, u, n}_{1}\right) > 1 - \delta\!\left(\epsilon\right)$ by \eqref{inequ:Prop_8_1} and
\begin{equation*}
 A^{\epsilon, u, n}_{1} \subseteq \left\{\omega \in \Omega: \left\vert \log\!P\!\left(t,T_{n}\right) +  \left(T_{n} - t\right) \ell_{t}\right\vert \leq \epsilon \left(T_{n} - t\right) \right\}.
\end{equation*} 
Consequently for $n \geq N_{\epsilon}^{u}$ on $A^{\epsilon, u, n}_{1}$ we have
\begin{equation}\label{Prop_8_Inequ}
  \exp\!\left[- \left(\epsilon + \ell_{t}\right) \left(T_{n}\!- t\right)\right] \leq P\!\left(t,T_{n}\right) \leq \exp\!\left[ \left(\epsilon - \ell_{t}\right) \left(T_{n}\!- t\right)\right]
\end{equation}
for all $t\in \left[0,u\right]$ and since $\ell_{0} \leq \ell_{t} \leq \ell_{u}$ for all $t\in \left[0,u\right]$ by the DIR theorem (see for example \citet{article_Hubalek}), we have that 
for $n \geq N_{\epsilon}^{u}$ on $A^{\epsilon, u, n}_{1}$ it holds
\begin{equation}\label{Prop_8_Inequality}
  \exp\!\left[ - \left(\epsilon + \ell_{u}\right) \left(T_{n}\!- u\right)\right] \leq \sup_{0\leq s\leq t} P\!\left(s,T_{n}\right) \leq \exp\!\left[ \left(\epsilon - \ell_{0}\right) T_{n}\right].
\end{equation}
For $t\geq 0$ we define 
 \begin{equation}\label{equ:B_1}
 B_{1}\!\left(t\right) \colonequals \left\{\omega \in \Omega: \lim_{n \rightarrow \infty} \sup_{0\leq s\leq t} S_{n}\!\left(s\right) < +\infty\right\}.
  \end{equation}
We then obtain for $t < u$ and $n \geq N_{\epsilon}^{u}$
\small
\begin{align*}
 \mathbb{P}\!\left(B_{1}\!\left(t\right)\right) 
 & \overset{\phantom{\eqref{Prop_8_Inequality}}}{=} \mathbb{P}\!\left(\left\{\sup_{0\leq s\leq t} S_{N_{\epsilon}^{u}-1}\left(s\right) < +\infty\right\} \cap \left\{\lim_{n \rightarrow +\infty}\sup_{0\leq s\leq t} \sum_{i=N_{\epsilon}^{u}}^{n}\!P\!\left(s,T_{i}\right)< +\infty\right\}\right)\nonumber\\
 & \overset{\phantom{\eqref{Prop_8_Inequality}}}{=} \mathbb{P}\!\left(\lim_{n \rightarrow \infty}\sup_{0\leq s\leq t} \sum_{i=N_{\epsilon}^{u}}^{n}\!P\!\left(s,T_{i}\right) < +\infty\right) \nonumber\\
 & \overset{\phantom{\eqref{Prop_8_Inequality}}}{=} \mathbb{P}\!\left(\left.\lim_{n \rightarrow \infty}\sup_{0\leq s\leq t} \sum_{i=N_{\epsilon}^{u}}^{n}\!P\!\left(s,T_{i}\right) < +\infty \right| A_{1}^{\epsilon, u, n}\right) \mathbb{P}\!\left(A_{1}^{\epsilon, u, n}\right) \nonumber\\
 & \phantom{===} + \mathbb{P}\!\left(\left.\lim_{n \rightarrow \infty}\sup_{0\leq s\leq t} \sum_{i=N_{\epsilon}^{u}}^{n}\!P\!\left(s,T_{i}\right) < +\infty \right| \Omega\backslash A_{1}^{\epsilon, u, n}\right) \mathbb{P}\!\left(\Omega\backslash A_{1}^{\epsilon, u, n}\right) \nonumber\\ 
 & \overset{\phantom{\eqref{Prop_8_Inequality}}}{\geq} \mathbb{P}\!\left(\left.\lim_{n \rightarrow \infty}\sup_{0\leq s\leq t} \sum_{i=N_{\epsilon}^{u}}^{n}\!P\!\left(s,T_{i}\right) < +\infty \right| A_{1}^{\epsilon, u, n}\right) \mathbb{P}\!\left(A_{1}^{\epsilon, u, n}\right) \nonumber\\
 & \overset{\phantom{\eqref{Prop_8_Inequality}}}{\geq} \mathbb{P}\!\left(\left.\lim_{n \rightarrow \infty} \sum_{i=N_{\epsilon}^{u}}^{n}\! \sup_{0\leq s\leq t} P\!\left(s,T_{i}\right) < +\infty \right| A_{1}^{\epsilon, u, n}\right) \mathbb{P}\!\left(A_{1}^{\epsilon, u, n}\right) \nonumber\\
 & \overset{\eqref{Prop_8_Inequality}}{\geq} \mathbb{P}\!\left(\left.\lim_{n \rightarrow \infty} \sum_{i=N_{\epsilon}^{u}}^{n}\! \exp\!\left[\left(\epsilon - \ell_{0}\right)T_{i}\right] < +\infty \right| A_{1}^{\epsilon, u, n}\right) \mathbb{P}\!\left(A_{1}^{\epsilon, u, n}\right) \nonumber\\
 & \overset{\phantom{\eqref{Prop_8_Inequality}}}{\geq} \left(1 - \delta\!\left(\epsilon\right)\right) 
 \rightarrow 1
\end{align*}
\normalsize
for $\epsilon \rightarrow 0$ since it holds $\mathbb{P}$-a.s.
\begin{equation*}
\lim_{n \rightarrow \infty} \frac{\exp\left(- \ell_{0} T_{n+1} \right)}{\exp\left(- \ell_{0} T_{n} \right)} = \exp\!\left(-\ell_{0} \delta\right) \in \left(0,1\right),
 \end{equation*}
 which implies by the ratio test that $\lim_{n \rightarrow \infty} \sum_{i=0}^{n} \exp\!\left[\left(\epsilon - \ell_{0}\right)T_{i}\right] < +\infty$ 
 \newline $\mathbb{P}$-a.s.\,\,for $\epsilon \rightarrow 0$. That means, it holds $S_{n} \overset{n \rightarrow \infty}{\longrightarrow} S_{\infty}$ in ucp.
 \par
  Hence by Theorem \ref{Prop_5} (i) and (iii) we get for all $t \geq 0$ that $0 < R_{t} < +\infty$ $\mathbb{P}$-a.s.\,\,with
 \begin{equation*}
   R_{t} = \frac{1}{S_{\infty}\!\left(t\right)}.
  \end{equation*}
 The exploding long-term simple rate, $L = + \infty$, is a result of Proposition 5.4 of \citet{Brody_Hughston}.
\end{proof}

Now, let us investigate what happens to the long-term rates if the long-term yield either vanishes or explodes. 
We see that besides the asymptotic behaviour of the yield, information about the long-term zero-coupon bond price is needed to 
state the consequences on the other long-term rates. 
For the analysis of the cases when $\ell$ is negative, we refer to \citet{Thesis}.

\begin{prop}\label{Prop_12}
 Let $\ell_{t} = 0$ $\mathbb{P}$-a.s.\,\,for all $t\geq 0$.
 If  $P$ exists finitely with $\inf\nolimits_{0\leq s\leq t} P_{s} > 0$ $\mathbb{P}$-a.s.\,\,for all $t\geq 0$, then $R_{t} = 0$ and $L_{t} = 0$ $\mathbb{P}$-a.s.\,\,for all $t\geq 0$.
\end{prop}

\begin{proof}
From Corollary \ref{Lemma_3} follows that $S_{n} \overset{n \rightarrow \infty}{\longrightarrow} +\infty$ in ucp, hence 
by applying Theorem \ref{Prop_5} (ii) we get that $R_{t} = 0$ $\mathbb{P}$-a.s.\,\,for all $t\geq 0$.
\par
To show that the long-term simple rate vanishes $\mathbb{P}$-a.s., we prove that for all $t\geq 0$ it holds that 
$\mathbb{P}\!\left(B_{2}\!\left(t\right)\right) = 1$ with $B_{2}\!\left(t\right)$ defined for $t\geq 0$ as follows
\begin{equation*}
 B_{2}\!\left(t\right)\! \colonequals \!\left\{\omega \in \Omega: \lim_{n \rightarrow \infty} \sup_{0 \leq s\leq t} L\!\left(s,T_{n}\right) = 0\right\}.
\end{equation*} 
We have for all $t\geq 0$
\begin{align}\label{Prop_12_Equ_1}
 \mathbb{P}\!\left(B_{2}\!\left(t\right)\right) 
 & \overset{\eqref{Simple_def}}{=} \mathbb{P}\!\left(\lim_{n \rightarrow \infty}\sup_{0\leq s\leq t} \frac{1}{\left(T_{n}-s\right) P\!\left(s,T_{n}\right)} = 0\right) \nonumber\\
 & \overset{\phantom{\eqref{Simple_def}}}{\geq} \mathbb{P}\!\left(\lim_{n \rightarrow \infty} \frac{1}{\left(T_{n}-t\right)\inf\nolimits_{0\leq s\leq t} P_{s}} = 0\right)
= 1.\nonumber
\end{align}
\end{proof}

In the following, we investigate exploding long-term yields. 

\begin{Theorem}\label{Prop_14}
 If $\ell = +\infty$, then $0 < R_{t} < +\infty$ $\mathbb{P}$-a.s.\,\,for all $t\geq 0$ and $L = + \infty$.
\end{Theorem}

\begin{proof}
First, we show that $S_{n} \overset{n \rightarrow \infty}{\longrightarrow} S_{\infty}$ in ucp. 
We know by \eqref{UCP_Def_2_1} that for all $t\geq 0$ and all $\epsilon > 0$ it holds 
\begin{align*}
 \mathbb{P}\!\left(\inf_{0\leq s\leq t} \left\vert Y\!\left(s,T_{n}\right) \right\vert > \epsilon\right) 
 & \overset{\eqref{Yield_def_1}}{\geq} \mathbb{P}\!\left(\inf_{0\leq s\leq t} \left\vert \log\!P\!\left(s,T_{n}\right) \right\vert > \epsilon\, T_{n}\right)
  \overset{n \rightarrow \infty}{\longrightarrow} 1,
\end{align*}
i.e.\,\,for all $t\geq 0$ and all $\epsilon > 0$ there exists a $N^{t}_{\epsilon} \in \mathbb{N}$ such that for all $n\geq N^{t}_{\epsilon}$
\begin{equation}\label{equ:Prop_14_1}
 \mathbb{P}\!\left(\inf_{0\leq s\leq t} \left\vert \log\!P\!\left(s,T_{n}\right) \right\vert > \epsilon\, T_{n}\right) > 1 - \delta\!\left(\epsilon\right)
\end{equation}
with $\delta\!\left(\epsilon\right) \rightarrow 0$ for $\epsilon \rightarrow +\infty$. 
Define for $\epsilon > 0$, $u\geq 0$ and $n \in \mathbb{N}$
\begin{equation}\label{def:A_3}
 A_{2}^{\epsilon, u, n} \colonequals \left\{\omega \in \Omega: \inf_{0\leq s\leq u} \left\vert \log\!P\!\left(s,T_{n}\right) \right\vert > \epsilon\, T_{n} \right\}.
\end{equation}
Then for $n \geq N_{\epsilon}^{u}$, $t < u$
and $B_{1}\!\left(t\right)$ defined as in \eqref{equ:B_1}, we obtain
\begin{align*}
 \mathbb{P}\!\left(B_{1}\!\left(t\right)\right) 
 & \overset{\phantom{\eqref{def:A_3}}}{=} \mathbb{P}\!\left(\lim_{n \rightarrow \infty}\sup_{0\leq s\leq t} \sum_{i=N_{\epsilon}^{u}}^{n}\!P\!\left(s,T_{i}\right) < +\infty\right) \nonumber\\
 & \overset{\phantom{\eqref{def:A_3}}}{\geq} \mathbb{P}\!\left(\left.\lim_{n \rightarrow \infty} \sum_{i=N_{\epsilon}^{u}}^{n}\! \sup_{0\leq s\leq t} P\!\left(s,T_{i}\right) < +\infty \right| A_{2}^{\epsilon, u, n}\right) \mathbb{P}\!\left(A_{2}^{\epsilon, u, n}\right) \nonumber\\
 & \overset{\eqref{def:A_3}}{\geq} \mathbb{P}\!\left(\left.\lim_{n \rightarrow \infty} \sum_{i=N_{\epsilon}^{u}}^{n}\! \exp\!\left(-\epsilon\, T_{n} \right) < +\infty \right| A_{2}^{\epsilon, u, n}\right) \mathbb{P}\!\left(A_{2}^{\epsilon, u, n}\right) \nonumber\\
 & \overset{\phantom{\eqref{def:A_3}}}{\geq} \left(1 - \delta\!\left(\epsilon\right)\right)
  \rightarrow 1
\end{align*}
for $\epsilon \rightarrow +\infty$ due to the ratio test. 
That means $S_{n} \overset{n \rightarrow \infty}{\longrightarrow} S_{\infty}$ in ucp and consequently $0 < R_{t} < +\infty$ $\mathbb{P}$-a.s.\,\,for all $t\geq 0$ due to Theorem \ref{Prop_5} (i).
\par
 Proposition 5.4 of \citet{Brody_Hughston} leads to $L = + \infty$.
\end{proof}

The following table summarises the influence of the long-term yield on the long-term swap rate and long-term simple rate.

 \begin{table}[htpb]
 \centering
\begin{tabular}[c]{||c|c||c|c||}\hline
\small{If the long-term} & \small{With P} & \small{Then the long-term} & \small{Then the long-term} \\ 
\small{yield is} &  & \small{swap rate is} & \small{simple rate is}\\ \hline \hline
$\ell = 0$ & $0 < P < + \infty$ & $R = 0$ & $L = 0$\\ \hline
$\ell > 0$ & $P=0$ & $0 < R < + \infty$ & $L = + \infty$\\ \hline
$\ell= + \infty$ & $P=0$ & $0 < R < + \infty$ & $L = + \infty$\\ \hline
\end{tabular}
\caption{Influence of the long-term yield on long-term rates.}
\label{Table1}
 \end{table}
 \vspace{0.1cm}

\subsection{Influence of the Long-Term Swap Rate on Long-Term Rates}

After we investigated the influence of the long-term yield on the long-term swap rate and long-term simple rate, 
we are also interested in the other direction of this relation. 

\begin{prop}\label{Prop_18}
If $R_{t} = 0$ $\mathbb{P}$-a.s.\,\,for all $t\geq 0$, then $\ell_{t} \leq 0$ $\mathbb{P}$-a.s.\,\,for all $t\geq 0$. 
\end{prop}

\begin{proof}
First, we show that $S_{n} \overset{n \rightarrow \infty}{\longrightarrow} +\infty$ in ucp. 
For this, let us assume $S_{n}$ converges in ucp. 
Then, according to Theorem \ref{Prop_5} (i) it is $0 < R_{t}$ $\mathbb{P}$-a.s.\,\,for all $t\geq 0$, but this is a contradiction and therefore $S_{n}$ converges to $+ \infty$ in ucp. 
\par
Consequently $\ell_{t} \leq 0$ $\mathbb{P}$-a.s.\,\,for all $t\geq 0$ due to Theorems \ref{Prop_8} and \ref{Prop_14}.
\end{proof}

Now, we investigate the behaviour of the long-term rates if the long-term swap rate is strictly positive.

\begin{prop}\label{Prop_20}
If $0 < R_{t} < +\infty$ $\mathbb{P}$-a.s.\,\,for all $t\geq 0$, then $\ell_{t} \geq 0$ and $L_{t} > 0$ $\mathbb{P}$-a.s.\,\,for all $t\geq 0$.

\end{prop}

\begin{proof}
We know from Corollary \ref{Cor_7} that $R_{t}\leq 0$ $\mathbb{P}$-a.s.\,\,for all $t\geq 0$ if $S_{n}$ converges to $+\infty$ in ucp. 
Hence if $R_{t} > 0$ $\mathbb{P}$-a.s.\,\,for all $t\geq 0$, we have $S_{n} \overset{n \rightarrow \infty}{\longrightarrow} S_{\infty}$ in ucp.
Then, according to Propositions 3.2.3 and 3.2.9 of \citet{Thesis} it holds $\mathbb{P}$-a.s.\,\,$\ell_{t} \geq 0$ for all $t\geq 0$.
\par
Further, $L_{t} > 0$ $\mathbb{P}$-a.s.\,\,for all $t\geq 0$ is a consequence of Proposition 3.2.11 of \citet{Thesis}. 

\end{proof}
The only case left now is a strictly negative long-term swap rate.

\begin{prop}\label{Prop_22}
If $-\infty < R_{t} < 0$ $\mathbb{P}$-a.s.\,\,for all $t\geq 0$, then $\ell_{t} \leq 0$ and $L_{t} = 0$ $\mathbb{P}$-a.s.\,\,for all $t\geq 0$.
\end{prop}

\begin{proof}
First, we show that $S_{n} \overset{n \rightarrow \infty}{\longrightarrow} +\infty$ in ucp. 
We know from Theorem \ref{Prop_5} (i) that $R_{t} > 0$ $\mathbb{P}$-a.s.\,\,for all $t\geq 0$ if $S_{n}$ converges to $S_{\infty}$ in ucp, 
but this is a contradiction to $R_{t} < 0$ $\mathbb{P}$-a.s. for all $t\geq 0$. 
As a consequence of Theorems \ref{Prop_8} and \ref{Prop_14} it is $\ell_{t} \leq 0$ $\mathbb{P}$-a.s.\,\,for all $t\geq 0$. 
\par
Since  $S_{n} \overset{n \rightarrow \infty}{\longrightarrow} +\infty$ in ucp and $R<0$, by Theorem \ref{Prop_5} (ii) we get that $P$ cannot exist finitely, hence $L_{t} = 0$ $\mathbb{P}$-a.s.\,\,for all $t\geq 0$ because of Lemma \ref{Lemma_SimpleRate}. 
\end{proof}

In the table below we summarise the influence of the long-term swap rate on the other long-term rates by using the previous results as well as Lemma \ref{Lemma_SimpleRate}. 
Note, that $-\infty < R_{t} < +\infty$ $\mathbb{P}$-a.s.\,\,for all $t\geq 0$ by Theorem \ref{Prop_5} (iii). 
Hence, only three different cases have to be distinguished, 
$R_{t} = 0$, $0 < R_{t} < +\infty$, and $-\infty < R_{t} < 0$ $\mathbb{P}$-a.s.\,\,for all $t\geq 0$.

 \begin{table}[htpb]
 \centering
\begin{tabular}[c]{||c|c||c|c||}\hline
\small{If the long-term} & \small{With P} & \small{Then the long-term} & \small{Then the long-term} \\ 
\small{swap rate is} &  & \small{yield is} & \small{simple rate is}\\ \hline \hline
$R = 0$ & $0 \leq P < +\infty$ & $\ell \leq 0$ & $0 \leq L \leq +\infty$\\ \hline 
 $0 < R < +\infty$ & $P=0$ & $\ell \geq 0$ & $0 < L \leq +\infty$\\ \hline
$-\infty < R < 0$ & $P= +\infty$ & $\ell \leq 0$ & $L = 0$\\ \hline
\end{tabular}
\caption{Influence of the long-term swap rate on long-term rates.}
\label{Table2}
 \end{table}

\subsection{Influence of the Long-Term Simple Rate on Long-Term Rates}

Finally, we want to know about the influence of the long-term simple rate on long-term yields and long-term swap rates. 
Since $L_{t} \geq 0$ $\mathbb{P}$-a.s.\,\,for all $t\geq 0$, it is sufficient to investigate the three different cases, where 
$L_{t} = 0$, or $0 < L_{t} < +\infty$ $\mathbb{P}$-a.s.\,\,for all $t\geq 0$, or $L = +\infty$. 

\begin{Theorem}\label{Prop_24}
If $L_{t} \geq 0$ $\mathbb{P}$-a.s.\,\,for all $t\geq 0$, then $\ell_{t} \leq 0$  and $R_{t} \leq 0$  $\mathbb{P}$-a.s.\,\,for all $t\geq 0$.  
Furthermore, $R_{t} = 0$ $\mathbb{P}$-a.s.\,\,for all $t\geq 0$ if $P_{t} < +\infty$ $\mathbb{P}$-a.s.\,\,for all $t\geq 0$.
\end{Theorem}

\begin{proof}
For the sake of simplicity we prove the result for the case $L_{t}=0$ $\mathbb{P}$-a.s.\,\,for all $t\geq 0$. The general case for $0\leq L<+\infty$ is completely analogous.
First, we show that $S_{n} \overset{n \rightarrow \infty}{\longrightarrow} +\infty$ in ucp. 
We know that for all $t\geq 0$ and all $\epsilon > 0$ it holds 
\begin{align*}
 \mathbb{P}\!\left(\sup_{0\leq s\leq t} \left\vert L\!\left(s,T_{n}\right) \right\vert \leq \epsilon\right) 
&  \overset{n \rightarrow \infty}{\longrightarrow} 1,
\end{align*}
i.e.\,\,by \eqref{Simple_def} for all $t\geq 0$ and all $\epsilon > 0$ there exists $N_{\epsilon}^{t} \in \mathbb{N}$ such that for all $n \geq N_{\epsilon}^{t}$ 
\begin{equation}\label{inequ:Prop_24_1}
  \mathbb{P}\!\left(\sup_{0\leq s\leq t} \left\vert \frac{1}{T_{n}-s}\left(\frac{1}{P\!\left(s,T_{n}\right)}- 1\right) \right\vert \leq \epsilon \right) > 1 - \delta\!\left(\epsilon\right)
\end{equation}
with $\delta\!\left(\epsilon\right) \rightarrow 0$ for $\epsilon \rightarrow 0$. 
Define for $\epsilon > 0$, $u\geq 0$ and $n \in \mathbb{N}$
\begin{equation}\label{def:A_4}
 A^{\epsilon, u, n}_{3} \colonequals  \left\{\omega \in \Omega: \sup_{0\leq s\leq u} \left\vert \frac{1}{T_{n}-s}\left(\frac{1}{P\!\left(s,T_{n}\right)}- 1\right) \right\vert \leq \epsilon \right\}.
\end{equation}
Let us define for $t\geq 0$ 
 \begin{equation*}
  B_{3}\!\left(t\right) \colonequals \left\{\omega \in \Omega: \lim_{n \rightarrow \infty} \inf_{0\leq s\leq t} S_{n}\!\left(s\right) = +\infty\right\}.
 \end{equation*}
For $t < u$ and $n \geq N_{\epsilon}^{u}$ we then obtain
\begin{align*}
 \mathbb{P}\!\left(B_{3}\!\left(t\right)\right) 
  & \overset{\phantom{\eqref{def:A_4}}}{\geq} \mathbb{P}\!\left(\left.\lim_{n \rightarrow \infty} \sum_{i=N_{\epsilon}^{u}}^{n}\! \inf_{0\leq s\leq t} P\!\left(s,T_{i}\right) = +\infty \right| A_{3}^{\epsilon, u, n}\right) \mathbb{P}\!\left(A_{3}^{\epsilon, u, n}\right) \nonumber\\
  & \overset{\eqref{def:A_4}}{\geq} \mathbb{P}\!\left(\left.\lim_{n \rightarrow \infty} \sum_{i=N_{\epsilon}^{u}}^{n}\! \frac{1}{1 + \epsilon\,T_{i}} = +\infty \right| A_{3}^{\epsilon, u, n}\right) \mathbb{P}\!\left(A_{3}^{\epsilon, u, n}\right) \nonumber\\
  & \overset{\phantom{\eqref{def:A_4}}}{\geq} \left(1 - \delta\!\left(\epsilon\right)\right) 
 \rightarrow 1
\end{align*}
for $\epsilon \rightarrow 0$. That means $S_{n} \overset{n \rightarrow \infty}{\longrightarrow} +\infty$ in ucp and consequently 
$\ell_{t} \leq 0$ $\mathbb{P}$-a.s.\,\,for all $t\geq 0$ due to Theorems \ref{Prop_8} and \ref{Prop_14}.
\par
The behaviour of the long-term swap rate is a direct consequence of Theorem \ref{Prop_5} (ii) and Corollary \ref{Cor_7}.
\end{proof}

Lastly, we are interested in the influence of an exploding long-term simple rate on the long-term yield and long-term swap rate.

\begin{Theorem}\label{Prop_28}
If $L = +\infty$, then $\ell_{t} \geq 0$ and $R_{t} > 0$ $\mathbb{P}$-a.s.\,\,for all $t\geq 0$.
\end{Theorem}

\begin{proof}
We show that $S_{n} \overset{n \rightarrow \infty}{\longrightarrow} S_{\infty}$ in ucp.
We know by \eqref{UCP_Def_2_1} that for all $t\geq 0$ and all $M > 0$
\begin{equation*}
\mathbb{P}\!\left(\inf_{0\leq s\leq t}  L\!\left(s,T_{n}\right)  > M\right)   \overset{n \rightarrow \infty}{\longrightarrow} 1.
\end{equation*}
Hence it holds by \eqref{Simple_def} for all $t\geq 0$ and all $\epsilon > 0$ that there exists $N_{\epsilon}^{t} \in \mathbb{N}$ such that for all $n \geq N_{\epsilon}^{t}$
\begin{equation}\label{inequ:Prop_28_1}
  \mathbb{P}\!\left( T_{n} \sup_{0\leq s\leq t}  P\!\left(s,T_{n}\right) \leq \epsilon \right) > 1 - \delta\!\left(\epsilon\right)
\end{equation}
with $\delta\!\left(\epsilon\right) \rightarrow 0$ for $\epsilon \rightarrow 0$.
Then, let us define for $\epsilon > 0$, $u\geq 0$ and $n \in \mathbb{N}$
\begin{equation}\label{def:A_6}
A^{\epsilon, u, n}_{5} \colonequals  \left\{\omega \in \Omega: T_{n} \sup_{0\leq s\leq u}  P\!\left(s,T_{n}\right) \leq \epsilon \right\}.
\end{equation}
For $t < u$ and $n \geq N_{\epsilon}^{u}$ we obtain by \eqref{inequ:Prop_28_1} with $B_{1}\!\left(t\right)$ defined as in \eqref{equ:B_1} that 
\begin{align*}
\mathbb{P}\!\left(B_{1}\!\left(t\right)\right)
  & \overset{\phantom{\eqref{def:A_6}}}{\geq} \mathbb{P}\!\left(\left.\lim_{n \rightarrow \infty} \sum_{i=N_{\epsilon}^{u}}^{n}\! \sup_{0\leq s\leq t} P\!\left(s,T_{i}\right) < +\infty \right| A_{5}^{\epsilon, u, n}\right) \mathbb{P}\!\left(A_{5}^{\epsilon, u, n}\right) \nonumber\\
 & \overset{\phantom{\eqref{def:A_6}}}{\geq} \left(1 - \delta\!\left(\epsilon\right)\right)  \rightarrow 1
\end{align*}
for $\epsilon \rightarrow 0$.
\end{proof}

Table \ref{Table3} summarises the influence of the long-term simple rate on the other long-term rates. 

 \begin{table}[htpb]
 \centering
\begin{tabular}[c]{||c|c||c|c||}\hline
\small{If the long-term} & \small{With P} & \small{Then the long-term} & \small{Then the long-term} \\ 
\small{simple rate is} &  & \small{yield is} & \small{swap rate is}\\ \hline \hline
$0 \leq L < +\infty$ & $0 \leq P < +\infty$ & $\ell \leq 0$ & $R = 0$\\ \hline
$0 \leq L < +\infty$ & $P = +\infty$ & $\ell \leq 0$ & $-\infty < R \leq 0$\\ \hline
$L = +\infty$ & $P = 0$ & $\ell \geq 0$ & $0 < R < +\infty$\\ \hline
\end{tabular}
\caption{Influence of the long-term simple rate on long-term rates.}
\label{Table3}
\end{table}

\section{Long-Term Rates in Specific Term Structure Models}\label{Specific_Term_Structure}

In this section we compute the long-term interest rates in two specific models, the Flesaker-Hughston model and the linear-rational model. 
We note that this class of models also includes affine interest rate models.
Let $\left(\Omega, \mathcal{F}, \mathbb{F}, \mathbb{P}\right)$ be the filtered probability space introduced in Section \ref{Rates}.

\subsection{Long-Term Rates in the Flesaker-Hughston Model}\label{Fle_Hugh}

We now derive the long-term swap rate in the Flesaker-Hughston interest rate model. 
The model has been introduced in \citet{Flesaker_Hughston1996} and further developed in \citet{article_MusielaRutkowski1997} and \citet{Rutkowski1997}, see also \citet{Rogers}. 
Main advantages of this approach are that it specifies non-negative interest rates only and has a high degree of tractability. 
Another appealing feature is that besides relatively simple models for bond prices, short and forward rates, there are closed-form formulas for caps, floors and swaptions available. 
In the following, we first shortly outline the generalised Flesaker-Hughston model that is explained in detail in \citet{Rutkowski1997} and then consider two specific cases. 
The basic input of the model is a strictly positive supermartingale $A$ on $\left(\Omega, \mathcal{F}, \mathbb{F}, \mathbb{P}\right)$ 
which represents the state price density, so that the zero-coupon bond price can be expressed as 
\begin{equation}\label{FH_Equ_1}
 P\!\left(t,T\right) = \frac{\mathbb{E}^{\mathbb{P}}\!\left[\left.A_{T}\right| \mathcal{F}_{t}\right]}{A_{t}}\,,\ 0 \leq t \leq T\,,
\end{equation}
for all $T \geq 0$.
It immediately follows $P\!\left(T,T\right) = 1$ for all $T \geq 0$ and $P\!\left(t,U\right) \leq P\!\left(t,T\right)$ for all $0 \leq t \leq T \leq U$, 
i.e.\,\,the zero-coupon bond price is a decreasing process in the maturity. 
This choice guarantees positive forward and short rates for all maturities (cf.\,\,equations (10) and (11) of \citet{Flesaker_Hughston1996}). 
To model the long-term yield and swap rate in this methodology a specific choice of $A$ is needed. 
For this matter, we focus on two special cases presented in Section 2.3 of \citet{Rutkowski1997}.

\begin{example}

The supermartingale $A$ is given by 
\begin{equation}\label{FH_Equ_2}
 A_{t} = f\!\left(t\right) + g\!\left(t\right) M_{t}\,,\ t \geq 0\,,
\end{equation}
where $f,g : \mathbb{R}_{+} \rightarrow \mathbb{R}_{+}$ are strictly positive decreasing functions 
and $M$ is a strictly positive martingale defined on $\left(\Omega, \mathcal{F}, \mathbb{F}, \mathbb{P}\right)$, with $M_{0} = 1$. 
We shall consider in the sequel a c\`{a}dl\`{a}g version of $M$.
It follows from \eqref{FH_Equ_1} that for all $0 \leq t \leq T$
\begin{equation}\label{FH_Equ_3}
 P\!\left(t,T\right) = \frac{f\!\left(T\right) + g\!\left(T\right) M_{t}}{f\!\left(t\right) + g\!\left(t\right) M_{t}}\,.
\end{equation}
The initial yield curve can easily be fitted by choosing strictly positive decreasing functions $f$ and $g$ in such a way that 
\begin{equation}\label{eq_bond_0}
 P\!\left(0,T\right) = \frac{f\!\left(T\right) + g\!\left(T\right)}{f\!\left(0\right) + g\!\left(0\right)}
\end{equation}
for all $T \geq 0$.
\par
For the calculations of the long-term yield and swap rate, we assume that the following conditions on the asymptotic behaviour of $f$ and $g$ hold:
\begin{equation}\label{FH_Equ_4}
F \colonequals \sum_{i=1}^{\infty} f\!\left(T_{i}\right) < +\infty\,,\ \ \ \ \ \ G \colonequals \sum_{i=1}^{\infty} g\!\left(T_{i}\right) < +\infty\,,
\end{equation}
with $F,G \in \mathbb{R}_{+}$.
\par
From \eqref{FH_Equ_4} it follows immediately that $\lim_{t \rightarrow \infty} f\!\left(t\right) = \lim_{t \rightarrow \infty} g\!\left(t\right) = 0$ 
and hence $P_{t} = 0$ for all $t\geq 0$. 
This condition is assumed in \citet{Flesaker_Hughston1996}, whereas here it follows from \eqref{FH_Equ_4}. 
We also get for all $t\geq 0$
\begin{equation*}
S_{\infty}\!\left(t\right) = \delta \, \frac{F + G M_{t}}{f\!\left(t\right) + g\!\left(t\right) M_{t}}\,\, \text{ $\mathbb{P}$-a.s.}
\end{equation*}
since 
\begin{align*}
 & \sup_{0\leq s\leq t} \left\vert \sum_{i=1}^{n} \frac{f\!\left(T_{i}\right) + g\!\left(T_{i}\right) M_{s}}{f\!\left(s\right) + g\!\left(s\right) M_{s}} - \frac{F + G M_{s}}{f\!\left(s\right) + g\!\left(s\right) M_{s}}\right\vert \\
 & \phantom{===} = \sup_{0\leq s\leq t} \left\vert \frac{M_{s}}{f\!\left(s\right) + g\!\left(s\right) M_{s}} \left(\sum_{i=1}^{n}\!g\!\left(T_{i}\right) - G\right) + \frac{\sum_{i=1}^{n}\!f\!\left(T_{i}\right) - F}{f\!\left(s\right) + g\!\left(s\right) M_{s}} \right\vert \\
 & \phantom{===} \rightarrow 0\,\, \text{ $\mathbb{P}$-a.s.}
\end{align*}
for all $t\geq 0$, hence in probability because 
\begin{equation*}
\sup_{0\leq s\leq t} \frac{M_{s}}{f\!\left(s\right) + g\!\left(s\right) M_{s}} \leq \sup_{0\leq s\leq t} \frac{M_{s}}{g\!\left(s\right) M_{s}} \leq \frac{1}{g\!\left(t\right)} < +\infty.
\end{equation*}
Then, by Proposition \ref{Prop_5} it holds $\mathbb{P}$-a.s.
\begin{equation}\label{FH_Equ_6}
R_{t} = \frac{f\!\left(t\right) + g\!\left(t\right) M_{t}}{\delta \left(F + G M_{t}\right)},\,\, t\geq 0.
\end{equation}
Now, we also want to compute the long-term yield in this model specification. It is for all $t\geq 0$
\begin{equation}\label{FH_Equ_7}
\ell_{\cdot} \overset{\eqref{Long_Yield}}{=} \lim_{T \rightarrow \infty} Y\!\left(\,\cdot\,,T\right) \overset{\eqref{Yield_def_1}}{\underset{\eqref{FH_Equ_3}}{=}} 
- \lim_{T \rightarrow \infty} T^{-1} \log\! \left(f\!\left(T\right) + g\!\left(T\right) M_{\cdot}\right)\,\,\text{ in ucp.}
\end{equation}
We know from Proposition \ref{Prop_20} that $\ell_{t} \geq 0$ $\mathbb{P}$-a.s.\,\,for all $t\geq 0$ since the long-term swap rate is strictly positive due to \eqref{FH_Equ_6} and $P$ vanishes. 
\par
Let us consider a simple example, where $f\!\left(t\right) = \exp\!\left(- \alpha t\right)$, $g\!\left(t\right) = \exp\!\left(- \beta t\right)$ with $0 < \alpha < \beta$. 
Then $f$ and $g$ are decreasing strictly positive functions and the ratio test shows that the infinite sums of $f$ and $g$  exist. 
We denote them by 
\begin{equation*}
\alpha_{\infty} \colonequals \sum_{i=1}^{\infty} \exp\!\left(- \alpha T_{i}\right) \,,\ \ \ \ \ \ \beta_{\infty} \colonequals \sum_{i=1}^{\infty} \exp\!\left(- \beta T_{i}\right)
\end{equation*}
with $0 < \beta_{\infty} \leq \alpha_{\infty}$.
Hence all required conditions are fulfilled and we get the following equations for the long-term swap rate and the long-term yield, respectively
\begin{equation*}
R_{t} \overset{\eqref{FH_Equ_6}}{=} \frac{\exp\!\left(- \alpha t\right) + \exp\!\left(- \beta t\right) M_{t}}{\delta \left(\alpha_{\infty} + \beta_{\infty} M_{t}\right)},\,\, t\geq 0,
\end{equation*}
and
\begin{align*}
\ell_{\cdot} & \overset{\eqref{FH_Equ_7}}{=} - \lim_{T \rightarrow \infty} T^{-1} \log\! \left(f\!\left(T\right) + g\!\left(T\right) M_{\cdot}\right) \nonumber\\
& \overset{\phantom{\eqref{FH_Equ_7}}}{=} - \lim_{T \rightarrow \infty} T^{-1} \log\! \left(\exp\!\left(-\alpha T\right)\left(1+ \exp\left(-\left(\beta - \alpha\right)T\right) M_{\cdot}\right)\right) \nonumber\\
&  \overset{\phantom{\eqref{FH_Equ_7}}}{=} \alpha + \lim_{T \rightarrow \infty} T^{-1} \log\! \left(1 + \exp\!\left(-\left(\beta - \alpha\right)T\right) M_{\cdot}\right) 
= \alpha\,\,\text{ in ucp.}
\end{align*}
It follows by Theorem \ref{Prop_8} that $L\!\left(t,T_{n}\right) \overset{n \rightarrow \infty}{\longrightarrow} +\infty$ in ucp. 
This result can also be obtained by direct computation since for all $t\geq 0$
\begin{equation*}
 \sup_{0\leq s\leq t} \frac{\exp\!\left(-\alpha s\right) + \exp\!\left(-\beta s\right) M_{s}}{T \exp\!\left(-\alpha T\right) + T \exp\!\left(-\beta T\right) M_{s}} \overset{T \rightarrow \infty}{\longrightarrow} +\infty\,\, \text{ $\mathbb{P}$-a.s.,}
\end{equation*}
i.e.\,\,in probability, since $M$ is c\`{a}dl\`{a}g.

\end{example}

\begin{example}

In the second special case of the Flesaker-Hughston model the supermartingale $A$ is defined as 
\begin{equation*}
 A_{t} = \int_{t}^{\infty}\!\phi\!\left(s\right) M\!\left(t,s\right) \, ds \,,\ t \geq 0\,,
\end{equation*}
where for every $s > 0$ the process $M\!\left(t,s\right), t \leq s$, is a strictly positive martingale on $\left(\Omega, \mathcal{F}, \mathbb{F}, \mathbb{P}\right)$ 
with $M\!\left(0,s\right) = 1$ such that $\int_{0}^{\infty}\!\phi\!\left(s\right) M\!\left(t,s\right) \, ds < +\infty$ $\mathbb{P}$-a.s.\,\,for all $t\geq 0$
and $\phi : \mathbb{R}_{+} \rightarrow \mathbb{R}_{+}$ is a strictly positive continuous function. 
From \eqref{FH_Equ_1} follows for $0\leq t\leq T$
\begin{equation}\label{FH_Equ2_2}
 P\!\left(t,T\right) = \frac{\int_{T}^{\infty}\!\phi\!\left(s\right) M\!\left(t,s\right) \, ds}{\int_{t}^{\infty}\!\phi\!\left(s\right) M\!\left(t,s\right) \, ds}\,,
\end{equation}
for all $T\geq 0$.
By differentiation of the zero-coupon bond price with respect to the maturity date, we see that the initial term structure satisfies 
$\phi\!\left(t\right) = - \frac{\partial P\left(0,t\right)}{\partial t}$ (cf.\,\,equation (6) of \citet{Flesaker_Hughston1996}). 
\par
According to \eqref{FH_Equ2_2} we get that $P_{t} = 0$ $\mathbb{P}$-a.s.\,\,for all $t \geq 0$.
We define $Q_{n} \colonequals \left(Q_{n}\!\left(t\right)\right)_{t\geq 0}$ for all $n\geq 0$ with 
\begin{equation*}
Q_{n}\!\left(t\right) \colonequals \sum_{i = 1}^{n} \! \int_{T_{i}}^{\infty}\!\phi\!\left(s\right) M\!\left(t,s\right) \, ds
\end{equation*}
and assume that for $Q \colonequals \left(Q\!\left(t\right)\right)_{t\geq 0}$ we have
\begin{equation*}
Q\!\left(t\right) \colonequals \sum_{i = 1}^{\infty} \! \int_{T_{i}}^{\infty}\!\phi\!\left(s\right) M\!\left(t,s\right) \, ds\, < +\infty
\end{equation*}
for all $t\geq 0$, and that $Q_{n} \overset{n \rightarrow \infty}{\longrightarrow} Q$ in ucp.
Then, we get $S_{n} \overset{n \rightarrow \infty}{\longrightarrow} S_{\infty} < +\infty$ in ucp and 
hence the convergences of the long-term swap rate and the long-term yield hold also in ucp.
Due to Theorem \ref{Prop_5} (i) the long-term swap rate is 
\begin{equation}\label{FH_Equ2_4}
R_{t} = \frac{\int_{t}^{\infty}\!\phi\!\left(s\right) M\!\left(t,s\right) \, ds}{\delta\,\sum_{i = 1}^{\infty} \! \int_{T_{i}}^{\infty}\!\phi\!\left(s\right) M\!\left(t,s\right) \, ds}\,, t\geq 0\,.
\end{equation}
Now, we again want to know the long-term yield in this case. 
It holds
\begin{equation*}
\ell_{\cdot} = - \lim_{T \rightarrow \infty} T^{-1} \log\!\left(\int_{T}^{\infty}\!\phi\!\left(s\right) M\!\left(\,\cdot\,,s\right) \, ds\right)\,\, \text{ in ucp.}
\end{equation*}
From Proposition \ref{Prop_20} we know that $\ell_{t} \geq 0$ $\mathbb{P}$-a.s.\,\,for all $t\geq 0$ since $R_{t} > 0$ $\mathbb{P}$-a.s.\,\,for all $t\geq 0$ due to \eqref{FH_Equ2_4} and $P$ vanishes. 
Further, $L_{t} \geq 0$ $\mathbb{P}$-a.s.\,\,for all $t\geq 0$ by Proposition \ref{Prop_20}.

\end{example}

\subsection{Long-Term Rates in the Linear-Rational Methodology}

The class of linear-rational term structure models is introduced in \citet{Filipovic2014} for the first time. 
This class presents several advantages: it is highly tractable and offers a very good fit to interest rate swaps and swaptions data since 1997. 
Further, non-negative interest rates are guaranteed, unspanned factors affecting volatility and risk premia are accommodated, and analytical solutions to swaptions are admitted. 
\par
We assume the existence of a state price density, i.e.\,\,of a positive adapted process $A \colonequals \left(A_{t}\right)_{t\geq 0}$ on $\left(\Omega, \mathcal{F}, \mathbb{F}, \mathbb{P}\right)$ 
such that the price $\Pi\!\left(t,T\right)$ at time $t$ of any time $T$ cashflow $C_{T}$ is given by 
\begin{equation}\label{LR_Equ_1}
 \Pi\!\left(t,T\right) = \frac{\mathbb{E}^{\mathbb{P}}\!\left[\left.A_{T} C_{T}\right| \mathcal{F}_{t}\right]}{A_{t}}\,,\ 0 \leq t \leq T\,,
\end{equation}
for all $T \geq 0$.
In particular we suppose that the state price density $A$ is driven by a multivariate factor process $X \colonequals \left(X_{t}\right)_{t\geq 0}$ with state space $E \subseteq \mathbb{R}^{d}$, $d \geq 1$, where 
\begin{equation}\label{LR_Equ_2}
 dX_{t} = k \left(\theta - X_{t}\right) dt + dM_{t}\,,\ t\geq 0\,,
\end{equation}
for some $k \in \mathbb{R}_{+}$, $\theta \in \mathbb{R}^{d}$, and some martingale $M \colonequals \left(M_{t}\right)_{t\geq 0}$ on $E$. 
We assume to work with the c\`{a}dl\`{a}g version of $X$. 
Next, $A$ is defined  as 
\begin{equation}\label{LR_Equ_3}
A_{t} \colonequals \exp\!\left(-\alpha t\right) \left(\phi + \psi^{\top}\!X_{t}\right)\,,\ t\geq 0\,,
\end{equation}
with $\phi \in \mathbb{R}$ and $\psi \in \mathbb{R}^{d}$ such that $\phi + \psi^{\top}\!x > 0$ for all $x \in E$, and $\alpha \in \mathbb{R}$. 
It holds $\alpha = \sup_{x \in E} \frac{k\,\psi^{\top}\left(\theta - x\right)}{\phi + \psi^{\top}\!x}$ to guarantee non-negative short rates (cf.\,\,equation (6) of \citet{Filipovic2014}).
Then, equations \eqref{LR_Equ_1}, \eqref{LR_Equ_2}, \eqref{LR_Equ_3}, together with the fact that $P\!\left(T,T\right) =1$ for all $T \geq 0$, lead to
\begin{equation}\label{LR_Equ_4}
 P\!\left(t,T\right) = \frac{\left(\phi + \psi^{\top}\theta\right)\exp\!\left(-\alpha\left(T\!-\!t\right)\right) + \psi^{\top}\!\left(X_{t} - \theta\right)\exp\!\left(-\left(\alpha\! +\! k\right)\left(T\!-\!t\right)\right)}{\phi + \psi^{\top}X_{t}}
\end{equation}
for all $0 \leq t \leq T$.
Hence, $P_{t} = 0$ $\mathbb{P}$-a.s.\,\,for all $t\geq 0$ and we know by the ratio test that for all $t\geq 0$
\small
\begin{equation*}
\alpha_{\infty}\!\left(t\right) \colonequals \sum_{i=1}^{\infty} \exp\!\left(- \alpha \left(T_{i}\! -\! t\right)\right) < +\infty \,,\ \ \ \ \beta_{\infty}\!\left(t\right) \colonequals \sum_{i=1}^{\infty} \exp\!\left(- \left(\alpha + k\right)\left(T_{i}\! - \! t\right)\right) < +\infty\,.
\end{equation*}
\normalsize
Then for all $t\geq 0$ $\mathbb{P}$-a.s.
\begin{equation}\label{LR_Equ_5}
S_{\infty}\!\left(t\right) 
\overset{\eqref{LR_Equ_4}}{=} \frac{\left(\phi + \psi^{\top}\!\theta\right)\alpha_{\infty}\!\left(t\right) + \psi^{\top}\!\left(X_{t} - \theta\right) \beta_{\infty}\!\left(t\right)}{\phi + \psi^{\top}\!X_{t}} < +\infty\,.
\end{equation}
It follows by Proposition \ref{Prop_5} that for all $t\geq 0$ $\mathbb{P}$-a.s.
\begin{equation*}
 R_{t} \overset{\eqref{LR_Equ_4}}{\underset{\eqref{LR_Equ_5}}{=}} \frac{\phi + \psi^{\top}X_{t}}{\delta \left(\left(\phi + \psi^{\top}\!\theta\right) \alpha_{\infty}\!\left(t\right) + \psi^{\top}\!\left(X_{t} - \theta\right) \beta_{\infty}\!\left(t\right)\right)}\,.
\end{equation*}

Finally, we want to know the form of the long-term yield in the linear-rational term structure methodology. 
We define $y \colonequals \phi + \psi^{\top}\!\theta$ and see that
for all $t\geq 0$ holds
 \begin{equation*}
  \log\!\left[y + \psi^{\top}\left(\sup_{0\leq s\leq t} X_{s}\! -\! \theta\right)e^{-k\left(T - t\right)}\right] 
  \geq \log\!\left[y + \psi^{\top}\left(X_{t}\! -\! \theta\right)e^{-k\left(T - t\right)}\right]
 \end{equation*}
 as well as 
 \begin{equation*}
   \log\!\left[y + \psi^{\top}\left(X_{t}\! -\! \theta\right)e^{-k\left(T - t\right)}\right]
  \geq \log\!\left[y + \psi^{\top}\left(\inf_{0\leq s\leq t} X_{s}\! -\! \theta\right)e^{-k\, T}\right].
 \end{equation*}
This yields $\mathbb{P}$-a.s.\,\,for all $t\geq 0$
\small
\begin{align*}
 \sup_{0\leq s\leq t} \left\vert \alpha + \frac{\log\! P\!\left(s,T\right)}{T}\right\vert & \overset{\eqref{LR_Equ_4}}{=} 
 \sup_{0\leq s\leq t} \left\vert \alpha \frac{s}{T} + \frac{1}{T} \log\!\left[y + \psi^{\top}\left(X_{s}\! -\! \theta\right)e^{-k\left(T - s\right)}\right] \right\vert \nonumber\\
   & \overset{\phantom{\eqref{LR_Equ_4}}}{\leq} \sup_{0\leq s\leq t} \left\vert \alpha \frac{s}{T} \right\vert + 
   \sup_{0\leq s\leq t}  \frac{1}{T} \left\vert\log\!\left[ y + \psi^{\top}\left(X_{s}\! -\! \theta\right)e^{-k\left(T - s\right)}\right]\right\vert \nonumber\\
  & \overset{T \rightarrow \infty}{\longrightarrow} 0
\end{align*}
\normalsize
because $\sup_{0\leq s\leq t} X_{s} < \infty$ $\mathbb{P}$-a.s.\,\,for all $t\geq 0$ since $X$ is c\`{a}dl\`{a}g. 
 Hence, we have for all $t\geq 0$ that 
 $\lim_{T \rightarrow \infty} \sup_{0\leq s\leq t} Y\!\left(s,T\right) = \alpha$ $\mathbb{P}$-a.s., consequently in probability, 
 i.e.\,\,we get $\ell_{t} = \alpha$ $\mathbb{P}$-a.s.\,\,for all $t\geq 0$. 
 In case of $\alpha$ positive, the long-term simple rate explodes due to Theorem \ref{Prop_8}.

\section{Application: Valuation of the Non-Optional Component of a CoCo Bond}\label{cali}
In this section we present an application of our results on the long-term swap rate to evaluate the non-optional component of a CoCo bond. \newline
Several banks issued CoCo bonds with perpetual characteristics in recent years. They are perpetual in the sense that the time to maturity is unbounded if the option for conversion is not executed. Hence, this kind of financial product can be understood as a perpetual floating rate bond combined with an embedded option. For investors it is crucial to know the value of the option and the non-optional component to make an informed investment decision. We calibrate a model for the long-term swap rate, and use the resulting specification to compute the price of the perpetual floating rate bond corresponding to the non-optional component of the CoCo bond.
In particular, we  consider the CoCo bond with ISIN $XS1002801758$ issued by Barclays, see \citet{cocoTermSheet2}.\newline
In  the same setting as in Section \ref{Fle_Hugh}, we assume that the strictly positive martingale $M=(M_t)_{t\geq 0}$ in \eqref{FH_Equ_2} satisfies
\begin{equation}\label{M_spec}
dM_t= \sigma_t M_t dW_t, \qquad M_0=1,
\end{equation}
where $W=(W_t)_{t\geq 0}$ is a one-dimensional Brownian motion on $\left(\Omega, \mathcal{F}, \mathbb{F}, \mathbb{P}\right)$ and $\sigma=(\sigma_t)_{t\geq 0}$ is a deterministic continuous function such that  $\int_0^{+\infty} \sigma_s^2 ds < \infty$.
For the functions $f,g$ in \eqref{FH_Equ_3}, we set $f(t)=k_1e^{-\alpha t}, g(t)=k_2e^{-\beta t}$, $t\geq 0$. We consider $T_0,T_1,\cdots$ such that $T_i-T_{i-1}=\delta$ for all $i=1, \cdots$. In particular we choose here $\delta= 3$ months, which is a typical interval between maturities for swaps on real markets. 

As first step, we estimate a term structure for the discount factors. This is achieved by considering market data of overnight indexed swap on December $22^{nd}$, 2016. From the market prices of overnight indexed swap we bootstrap a term structure of discount factors by relying on the Finmath Java library (see \citet{finmath}). Secondly, we estimate the parameters for the functions $f,\  g,\ F,$ and $G$  in \eqref{eq_bond_0} by minimizing the squared distance between the term structure of zero coupon bonds obtained from the bootstrap and the right hand side of \eqref{eq_bond_0}. We obtain $k_1 =0.4894723$ and $\alpha=0.1536072$ for the function $f$, and $k_2= 8.6235042$ and $\beta=0.0117588$ for the function $g$. By using this result, we compute $F=165.95163$ and $G=11742.367$ by evaluating the sums  in 
\eqref{FH_Equ_4} along a time discretization with a horizon of 1000 years.\\
Concerning the volatility function $\sigma$ in \eqref{M_spec}, we set $\sigma_t=e^{-\lambda t}$.  For the estimation of the parameter $\lambda$, we rely on Remark~\ref{consol}. Since  the long term swap rate is proportional to the consol rate of a perpetual bond, we  use time series data of a consol bond to estimate the missing volatility parameter. We consider the yield of the perpetual Bond with Isin $BMG7498P3093$ and perform a maximum likelihood estimation, obtaining $\lambda=0.0748829$.

With the given full specification of the process $R$, we then estimate the value of the non-optional part of the CoCo Bond $XS1002801758$. Let $P_{NO}(T_0)$ denote the price of the non-optional part at time $T_0$. From the term sheet of the CoCo, we observe that the investor initially receives an $8\%$ fixed coupon up to 2020, where no conversion is possible. On the other side, when the claim starts exhibiting an optionality feature, the investor receives $6.75\%$ plus the mid market swap rate. 

A simple estimate for the non-optional component of the stream of payments involved in the CoCo is simply given by considering the following perpetual bond with unit notional:

\begin{align*}
P_{NO}(T_0)=\sum_{i=0}^{\mathcal{N}}P(T_0,T_i)\left(R_{T_i}+S\right)
\end{align*}
where we truncate the infinite sum up to $\mathcal{N}=50$. The long term estimate of the Euribor-Eonia spread $S= 0.0011$ is obtained from the bootstrapped curves by assuming a constant Euribor-Eonia Spread for maturities larger than 50 years.  Using \eqref{FH_Equ_6} we simulate the dynamics of the long term swap rate $R$ and obtain $P_{NO}(T_0)=0.1969$ by a Monte Carlo simulation. By using  the market price $P_{MKT}(T_0)=1.05386$  of $XS1002801758$ observed on December $22$, 2016,  we immediately deduce our estimate of the value of the embedded optional part, which is equal to $0.85695$. Such a result shows that the price of the CoCo bond is mainly driven by the embedded option. As the option is written on a unit notional, it can be concluded that the option tends to a position on the underlying.

%
%

\appendix

\section{Behaviour of $S_{\infty}$}\label{Infinite_Sum_of_Bond_Prices}

For the study of the long-term swap rate in Section \ref{Long_Term_Rates_Section} as well as 
of the relations among the different long-term interest rates in Section \ref{Relation_between_Long_Term_Rates} we need to obtain some results on the infinite sum $S_{\infty}$ of bond prices defined in \eqref{Sum_Infinite_Bond_Prices}. 
We recall that we consider a tenor structure with $c < \sup_{i \in \mathbb{N} \setminus \left\{0\right\}} (T_{i}-T_{i-1}) < C$ for some $c, C \in \mathbb{R}_{+}, c < C$.
The next two statements give insight about the relation between the long-term zero-coupon bond prices and the asymptotic behaviour of the sum of these prices, 
whereas Lemma \ref{Lemma_SimpleRate} tells us that the long-term simple rate vanishes if $P$ explodes.

\begin{prop}\label{Lemma_4}
If $S_{n} \overset{n \rightarrow \infty}{\longrightarrow} S_{\infty}$ in ucp, then $P_{t} = 0$ $\mathbb{P}$-a.s.\,\,for all $t\geq 0$.
\end{prop}

\begin{proof}
From $S_{n}\! \overset{n \rightarrow \infty}{\longrightarrow}\! S_{\infty}$ in ucp it follows that\! $S_{\infty}\!\left(t\right)\! <\! +\infty$ $\mathbb{P}$-a.s.\,for all $t\!\geq \!0$. 
We get for all $\epsilon > 0$ and $t \geq 0$ with $C^{\epsilon, t, n}\! \colonequals\! \left\{\omega \in \Omega: \sup_{0\leq s\leq t} \left\vert P\!\left(s,T_{n}\right) \right\vert \!> \!\epsilon\right\}$ 
\small
\begin{align*}
\mathbb{P}\!\left(C^{\epsilon, t, n}\right) 
& \overset{\phantom{(\ast)}}{\leq} \mathbb{P}\!\left(\sup_{0\leq s\leq t} \left\vert S_{n}\!\left(s\right) - S_{n\!-\!1}\!\left(s\right)\right\vert > \epsilon\, c\right) \nonumber\\
& \overset{\phantom{(\ast)}}{=} \mathbb{P}\!\left(\sup_{0\leq s\leq t} \left\vert S_{n}\!\left(s\right) - S_{\infty}\!\left(s\right) + S_{\infty}\!\left(s\right) - S_{n-1}\!\left(s\right)\right\vert > \epsilon\, c\right) \nonumber\\
& \overset{\phantom{(\ast)}}{\leq} \mathbb{P}\!\left(\sup_{0\leq s\leq t} \left(\left\vert S_{n}\!\left(s\right) - S_{\infty}\!\left(s\right)\right\vert + \left\vert S_{n-1}\!\left(s\right) - S_{\infty}\!\left(s\right)\right\vert\right) > \epsilon\, c\right) \nonumber\\
& \overset{\phantom{(\ast)}}{\leq} \mathbb{P}\!\left(\sup_{0\leq s\leq t} \left\vert S_{n}\!\left(s\right) - S_{\infty}\!\left(s\right)\right\vert + \sup_{0\leq s\leq t}\left\vert S_{n\!-\!1}\!\left(s\right) - S_{\infty}\!\left(s\right)\right\vert > \epsilon\, c\right) \nonumber\\
& \overset{\phantom{(\ast)}}{\leq}  \mathbb{P}\!\left(\left\{\sup_{0\leq s\leq t}\left\vert S_{n}\!\left(s\right) - S_{\infty}\!\left(s\right)\right\vert > \frac{\epsilon\, c}{2}\right\} \cup \left\{\sup_{0\leq s\leq t}\left\vert S_{n\!-\!1}\!\left(s\right) - S_{\infty}\!\left(s\right)\right\vert > \frac{\epsilon\, c}{2}\right\}\right) \nonumber\\
& \overset{(\ast)}{\leq} \mathbb{P}\!\left(\sup_{0\leq s\leq t} \left\vert S_{n}\!\left(s\right) - S_{\infty}\!\left(s\right)\right\vert > \frac{\epsilon\, c}{2}\right) + \mathbb{P}\!\left(\sup_{0\leq s\leq t} \left\vert S_{n\!-\!1}\!\left(s\right) - S_{\infty}\!\left(s\right)\right\vert > \frac{\epsilon\, c}{2}\right) \nonumber\\
& \overset{n \rightarrow \infty}{\longrightarrow} 0
\end{align*}
\normalsize
due to the ucp convergence of $S_{n}$.
Hence $P_{t} = 0$ $\mathbb{P}$-a.s.\,\,for all $t\geq 0$.
\end{proof}

\begin{Corollary}\label{Lemma_3}
If $\mathbb{P}\!\left(P_{t} > 0\right) > 0$ for some $t\geq 0$, then $S_{n} \overset{n \rightarrow \infty}{\longrightarrow} +\infty$ in ucp.
\end{Corollary}

\begin{proof}
This is a direct consequence of Proposition \ref{Lemma_4}.
\end{proof}

\begin{Lemma}\label{Lemma_SimpleRate}
If $P = +\infty$, it follows $L_{t} = 0$ $\mathbb{P}$-a.s.\,\,for all $t\geq 0$.
\end{Lemma}

\begin{proof}
It follows $L\!\left(\,\cdot\,, T_{n}\right) \overset{n \rightarrow \infty}{\longrightarrow} 0$ in ucp 
by \eqref{Simple_def} and the definition of convergence to $+ \infty$ in ucp (cf.\,\,Definition \ref{UCP_Def_2}). 
\end{proof}

\section{UCP Convergence}\label{appendix_ucp}

The definition of uniform convergence on compacts in probability (ucp convergence) can be found in Chapter II, Section 4 of \citet{Book_Protter}. 
We repeat this here for the reader's convenience. As before we consider a stochastic basis $\left(\Omega, \mathcal{F}, \mathbb{P}\right)$ endowed 
with the filtration $\mathbb{F} \colonequals \left(\mathcal{F}_{t}\right)_{t\geq 0}$ with $\mathcal{F}_{\infty} \subseteq \mathcal{F}$ 
satisfying the usual hypothesis. All processes are adapted to $\mathbb{F}$.

\begin{Def}\label{UCP_Def_1}
 A sequence of processes $\left(X^{n}\right)_{n \in \mathbb{N}}$ with values in $\mathbb{R}^d$ converges to a process $X$ uniformly on compacts in probability if, for each $t > 0$, 
 $\sup_{0\leq s\leq t} \|X_{s}^{n} - X_{s} \|$ converges to $0$ in probability, i.e.\,\,for all $\epsilon > 0$  it holds
 \begin{equation}\label{UCP_Def_1_1}
   \mathbb{P}\!\left(\sup_{0\leq s\leq t} \|X_{s}^{n} - X_{s}\|> \epsilon\right) \overset{n \rightarrow \infty}{\longrightarrow} 0\,.
 \end{equation}
We write $X^{n} \overset{n \rightarrow \infty}{\longrightarrow} X$ in ucp.
\end{Def}

\begin{Theorem}\label{UCP_Thm_1}
  Let $\left(X^{n}\right)_{n \in \mathbb{N}}$ and $\left(Y^{n}\right)_{n \in \mathbb{N}}$ be sequences of real-valued processes .
  If $\left(X^{n},Y^{n}\right)\! \overset{n \rightarrow \infty}{\longrightarrow}\! \left(X,Y\right)$ in ucp with 
  $\sup_{0 \leq s \leq t}\! \left\vert X_{s} \right\vert\! <\! +\infty$ and $\sup_{0 \leq s \leq t}\! \left\vert Y_{s} \right\vert\! <\! +\infty$ $\mathbb{P}$-a.s.\,\,for all $t\geq 0$, then 
  $f\!\left(X^{n},Y^{n}\right) \overset{n \rightarrow \infty}{\longrightarrow} f\!\left(X,Y\right)$ in ucp for all 
  $f : \mathbb{R}^{2} \rightarrow \mathbb{R}$ continuous.
\end{Theorem}

\begin{proof}
 Let us define $\nu_{s}^{n} \colonequals \left(X_{s}^{n},Y_{s}^{n}\right)$, $\nu_{s} \colonequals \left(X_{s},Y_{s}\right)$, and let
 $\left\Vert \cdot \right\Vert$ be the Euclidean norm on $\mathbb{R}^{2}$. We have to show that for all $t\geq 0$ and $\epsilon > 0$ it holds
 \begin{equation}\label{UCP_Thm_1_1}
  \mathbb{P}\!\left(\sup_{0 \leq s \leq t} \left\vert f\!\left(\nu_{s}^{n}\right) - f\!\left(\nu_{s}\right) \right\vert > \epsilon\right) \overset{n \rightarrow \infty}{\longrightarrow} 0.
 \end{equation}
 Let $k \geq 0$. Then for all $t\geq 0$ it holds
 \begin{align}\label{UCP_Thm_1_2}
  \left\{\sup_{0\leq s\leq t} \left\vert f\!\left(\nu^{n}_{s}\right) - f\!\left(\nu_{s}\right)\right\vert > \epsilon\right\} & \subseteq
  \left\{\sup_{0\leq s\leq t} \left\vert f\!\left(\nu^{n}_{s}\right) - f\!\left(\nu_{s}\right)\right\vert > \epsilon, \sup_{0\leq s\leq t} \left\Vert \nu_{s} \right\Vert \leq k\right\} \nonumber\\
     & \phantom{===} \cup \left\{\sup_{0\leq s\leq t} \left\Vert \nu_{s} \right\Vert > k\right\}.
 \end{align}
By the Heine-Cantor theorem (cf.\,\,Theorem A.1.1 of \citet{Book_Canuto_Tabacco}) it follows from $f$ continuous that $f$ is uniformly continuous on any bounded interval and therefore there exists for the given $\epsilon > 0$ a $\delta > 0$ such that
\small
 \begin{align}\label{UCP_Thm_1_3}
  \left\{\sup_{0\leq s\leq t}\! \left\vert f\!\left(\nu^{n}_{s}\right) - f\!\left(\nu_{s}\right)\right\vert > \epsilon, \sup_{0\leq s\leq t}\! \left\Vert \nu_{s} \right\Vert \leq k\right\}
     & \subseteq \left\{\sup_{0\leq s\leq t}\! \left\Vert \nu^{n}_{s} - \nu_{s}\right\Vert > \delta, \sup_{0\leq s\leq t}\! \left\Vert \nu_{s} \right\Vert \leq k\right\} \nonumber\\
     & \subseteq \left\{\sup_{0\leq s\leq t}\! \left\Vert \nu^{n}_{s} - \nu_{s}\right\Vert > \delta\right\}
 \end{align}
 \normalsize
Substituting \eqref{UCP_Thm_1_3} into \eqref{UCP_Thm_1_2} gives us 
\small
\begin{equation}\label{UCP_Thm_1_4}
 \left\{\sup_{0\leq s\leq t}\! \left\vert f\!\left(\nu^{n}_{s}\right) - f\!\left(\nu_{s}\right)\right\vert > \epsilon\right\} \subseteq 
 \left\{\sup_{0\leq s\leq t}\! \left\Vert \nu^{n}_{s} - \nu_{s}\right\Vert > \delta\right\} \cup \left\{\sup_{0\leq s\leq t}\! \left\Vert \nu_{s} \right\Vert > k\right\}.\nonumber
\end{equation}
\normalsize
Hence
\small
\begin{equation}\label{UCP_Thm_1_5}
 \mathbb{P}\!\left(\sup_{0\leq s\leq t}\! \left\vert f\!\left(\nu^{n}_{s}\right) - f\!\left(\nu_{s}\right)\right\vert > \epsilon\right)  \leq 
 \mathbb{P}\!\left(\sup_{0\leq s\leq t}\! \left\Vert \nu^{n}_{s} - \nu_{s}\right\Vert > \delta\right) 
 + \mathbb{P}\!\left(\sup_{0\leq s\leq t}\! \left\Vert \nu_{s} \right\Vert > k\right).
\end{equation}
\normalsize
Since $\sup_{0 \leq s \leq t} \left\vert X_{s} \right\vert < +\infty$ and $\sup_{0 \leq s \leq t} \left\vert Y_{s} \right\vert < +\infty$ $\mathbb{P}$-a.s.\,\,for all $t\geq 0$, 
it holds for all $t\geq 0$ that $\mathbb{P}\!\left(\sup_{0\leq s\leq t} \left\Vert \nu_{s} \right\Vert > k\right) \overset{k \rightarrow \infty}{\longrightarrow} 0$. 
Let first $k \rightarrow \infty$ and then $n \rightarrow \infty$, to obtain \eqref{UCP_Thm_1_1} from \eqref{UCP_Thm_1_5}.
\end{proof}

In order to treat the case of exploding long-term interest rates, we now provide a definition of convergence to $\pm \infty$ in ucp. 

\begin{Def}\label{UCP_Def_2}
  A sequence of real-valued processes $\left(X^{n}\right)_{n \in \mathbb{N}}$ converges to $+ \infty$ uniformly on compacts in probability if, for each $t > 0$ and $M > 0$ 
  it holds 
  \begin{equation}\label{UCP_Def_2_1}
 \mathbb{P}\!\left(\inf_{0\leq s\leq t} X_{s}^{n} > M\right) \overset{n \rightarrow \infty}{\longrightarrow} 1\,.
\end{equation}
We write $X^{n} \overset{n \rightarrow \infty}{\longrightarrow} + \infty$ in ucp. 
\par
Accordingly the sequence of real-valued processes $\left(X^{n}\right)_{n \in \mathbb{N}}$ converges to $- \infty$ uniformly on compacts in probability if, for each $t > 0$ and $M > 0$ 
  it holds 
\begin{equation}\label{UCP_Def_2_2}
 \mathbb{P}\!\left(\sup_{0\leq s\leq t} X_{s}^{n} < -M\right) \overset{n \rightarrow \infty}{\longrightarrow} 1\,.
\end{equation}
Then, we write $X^{n} \overset{n \rightarrow \infty}{\longrightarrow} -\infty$ in ucp. 
\end{Def}

\section*{Acknowledgements}

The authors wish to thank Damir Filipovi\'{c}, Paolo Guasoni, Alexander Lipton, Chris Rogers, and Wolfgang Runggaldier 
for their interesting remarks on this paper during the 7th General AMaMeF and Swissquote Conference, held at EPFL in Lausanne in September 2015. \\
We also thank Tomislav Maras for his help with some simulations.
\bibliography{biblio}

\begin{thebibliography}{40}
\providecommand{\natexlab}[1]{#1}
\providecommand{\url}[1]{\texttt{#1}}
\expandafter\ifx\csname urlstyle\endcsname\relax
  \providecommand{\doi}[1]{doi: #1}\else
  \providecommand{\doi}{doi: \begingroup \urlstyle{rm}\Url}\fi

\bibitem[Albul et~al.(2010)Albul, Jaffee, and Tchistyi]{Coco_1}
B.~Albul, D.M. Jaffee, and A.~Tchistyi.
\newblock {Contingent Convertible Bonds and Capital Structure}.
\newblock \emph{{Coleman Fung Risk Management Research Center Working Papers
  2006 - 2013}}, 01, 2010.

\bibitem[Bernanke(2009)]{Coco_6}
B.S. Bernanke.
\newblock {Regulatory Reform}.
\newblock \emph{Testimony before the Committee on Financial Services}, 2009.
\newblock {U.S. House of Representatives}.

\bibitem[Biagini and H{\"a}rtel(2014)]{bh2012}
F.~Biagini and M.~H{\"a}rtel.
\newblock {Behavior of {L}ong-{T}erm {Y}ields in a {L}{\'e}vy {T}erm
  {S}tructure}.
\newblock \emph{International Journal of Theoretical and Applied Finance},
  17\penalty0 (3):\penalty0 1--24, 2014.

\bibitem[Biagini et~al.(2018)Biagini, Gnoatto, and
  H{\"a}rtel]{BiaginiGnoattoHaertel}
F.~Biagini, A.~Gnoatto, and M.~H{\"a}rtel.
\newblock {Affine HJM Framework on $S_{d}^{+}$ and Long-Term Yield}.
\newblock \emph{Applied Mathematics and Optimization}, 77\penalty0
  (3):\penalty0 405--441, 2018.

\bibitem[{Board of Governors of the Federal Reserve System}(2015)]{FED2015}
{Board of Governors of the Federal Reserve System}.
\newblock {Selected Interest Rates (Daily)}.
\newblock \emph{Statistical Releases and Historical Data}, 2015.
\newblock {Available at Federal Reserve:
  http://www.federalreserve.gov/releases/h15/, Accessed: 2016-07-11}.

\bibitem[Brigo et~al.(2015)Brigo, Garcia, and Pede]{article_BrigoGarciaPede}
D.~Brigo, J.~Garcia, and N.~Pede.
\newblock {CoCo Bonds Pricing with Credit and Equity Calibrated First-Passage
  Firm Value Models}.
\newblock \emph{International Journal of Theoretical and Applied Finance},
  18\penalty0 (3):\penalty0 1--31, 2015.

\bibitem[Brody and Hughston(2016)]{Brody_Hughston}
D.~C. Brody and L.~P. Hughston.
\newblock Social discounting and the long rate of interest.
\newblock \emph{Mathematical Finance}, 28\penalty0 (1):\penalty0 306--334, May
  2016.

\bibitem[Canuto and Tabacco(2015)]{Book_Canuto_Tabacco}
C.~Canuto and A.~Tabacco.
\newblock \emph{{Mathematical Analysis II}}.
\newblock Springer, 1st edition, 2015.

\bibitem[Cuchiero et~al.(2016{\natexlab{a}})Cuchiero, Fontana, and
  Gnoatto]{cuchieroFontanaGnoatto}
C.~Cuchiero, C.~Fontana, and A.~Gnoatto.
\newblock {A General HJM Framework for Multiple Yield Curve Modeling}.
\newblock \emph{Finance and Stochastics}, 20\penalty0 (2):\penalty0 267 -- 320,
  2016{\natexlab{a}}.

\bibitem[Cuchiero et~al.(2016{\natexlab{b}})Cuchiero, Klein, and
  Teichmann]{Cuchiero2016}
C.~Cuchiero, I.~Klein, and J.~Teichmann.
\newblock A new perspective on the fundamental theorem of asset pricing for
  large financial markets.
\newblock \emph{Theory of Probability {\&} Its Applications}, 60\penalty0
  (4):\penalty0 561--579, January 2016{\natexlab{b}}.

\bibitem[Cuchiero et~al.(2018)Cuchiero, Fontana, and Gnoatto]{Cuchiero2018}
C.~Cuchiero, C.~Fontana, and A.~Gnoatto.
\newblock Affine multiple yield curve models.
\newblock \emph{Mathematical Finance}, 29\penalty0 (2):\penalty0 568--611,
  August 2018.

\bibitem[Delbaen(1993)]{Delsaen}
F.~Delbaen.
\newblock {Consols in the Cir Model}.
\newblock \emph{Mathematical Finance}, 3\penalty0 (2):\penalty0 125--134, 1993.

\bibitem[Dudley(2009)]{Coco_7}
W.C. Dudley.
\newblock {Some Lessons from the Crisis}.
\newblock \emph{Remarks to the Financial Crisis}, 2009.
\newblock {Institute of International Bankers Membership Luncheon}.

\bibitem[Duffie(2009)]{Coco_5}
D.~Duffie.
\newblock {Contractual Methods for Out-of-Court Restructuring of Systemically
  Important Financial Institutions}.
\newblock \emph{Submission Requested by the U.S. Treasury Working Group on Bank
  Capital}, 2009.

\bibitem[Duffie et~al.(1995)Duffie, Ma, and Yong]{DuffieJinMaYong}
D.~Duffie, J.~Ma, and J.~Yong.
\newblock {Black's Consol Rate Conjecture}.
\newblock \emph{The Annals of Applied Probability}, 5\penalty0 (2):\penalty0
  356--382, 1995.

\bibitem[Dybvig et~al.(1996)Dybvig, Ingersoll, and Ross]{dybvig96}
P.~H. Dybvig, J.~E. Ingersoll, and S.~A. Ross.
\newblock {Long Forward and Zero-Coupon Rates Can Never Fall}.
\newblock \emph{The Journal of Business}, 69\penalty0 (1):\penalty0 1--25,
  1996.

\bibitem[{El Karoui} et~al.(1997){El Karoui}, Frachot, and Geman]{Karoui97}
N.~El {El Karoui}, A.~Frachot, and H.~Geman.
\newblock {A Note on the Behavior of Long Zero Coupon Rates in a No Arbitrage
  Framework}.
\newblock \emph{Review of Derivatives Research}, 1\penalty0 (4):\penalty0
  351--369, 1997.

\bibitem[{European Central Bank}(2015)]{ECB2015}
{European Central Bank}.
\newblock {Long-Term Interest Rates for EU Member States}.
\newblock \emph{Monetary and Financial Statistics of ECB}, 2015.
\newblock URL \url{http://www.ecb.int/stats/money/long/html/index.en.html}.
\newblock {Available at ECB:
  http://www.ecb.int/stats/money/long/html/index.en.html, Accessed:
  2016-07-11}.

\bibitem[Filipovi{\'c} and Trolle(2013)]{FilipovicTrolle}
D.~Filipovi{\'c} and A.~B. Trolle.
\newblock {The Term Structure of Interbank Risk}.
\newblock \emph{Journal of Financial Economics}, 109\penalty0 (4):\penalty0
  707--733, 2013.

\bibitem[Filipovi{\'c} et~al.(2017)Filipovi{\'c}, Larsson, and
  Trolle]{Filipovic2014}
D.~Filipovi{\'c}, M.~Larsson, and A.~Trolle.
\newblock {Linear Rational Term-Structure Models}.
\newblock \emph{Journal of Finance}, 24:\penalty0 655 -- 704, 2017.

\bibitem[Flannery(2005)]{Coco_4}
M.J. Flannery.
\newblock {No Pain, No Gain: Effecting Market Discipline via Reverse
  Convertible Debentures}.
\newblock In \emph{{Capital Adequacy beyond Basel: Bank ing, Securities, and
  Insurance}}, chapter~5. Oxford University Press, 2005.

\bibitem[Flannery(2009)]{Coco_2}
M.J. Flannery.
\newblock {Stabilizing Large Financial Institutions with Contingent Capital
  Certificates}.
\newblock \emph{University of Florida, Working Paper}, 2009.
\newblock {Available at SSRN: http://ssrn.com/abstract=1485689, Accessed:
  2016-07-11}.

\bibitem[Flesaker and Hughston(1996)]{Flesaker_Hughston1996}
B.~Flesaker and L.P. Hughston.
\newblock {Positive Interest}.
\newblock \emph{Risk Magazine}, 9\penalty0 (1):\penalty0 46--49, 1996.

\bibitem[Fries(2018)]{finmath}
C.~Fries.
\newblock Finmath lib v3.1.3.
\newblock \url{http://www.finmath.net}, 2018.

\bibitem[Goldammer and Schmock(2012)]{article_Goldammer}
V.~Goldammer and U.~Schmock.
\newblock {Generalization of the Dybvig--Ingersoll--Ross Theorem and Asymptotic
  Minimality}.
\newblock \emph{Mathematical Finance}, 22\penalty0 (1):\penalty0 185--213,
  2012.

\bibitem[H{\"a}rtel(2015)]{Thesis}
M.~H{\"a}rtel.
\newblock \emph{{The Asymptotic Behavior of the Term Structure of Interest
  Rates}}.
\newblock PhD thesis, LMU University of Munich, 2015.

\bibitem[Henrard(2014)]{Book_Henrard}
M.~Henrard.
\newblock \emph{{Interest {R}ate {M}odelling in the {M}ulti-curve
  {F}ramework}}.
\newblock Palgrave Macmillan, 1st edition, 2014.

\bibitem[Hubalek et~al.(2002)Hubalek, Klein, and Teichmann]{article_Hubalek}
F.~Hubalek, I.~Klein, and J.~Teichmann.
\newblock {A General Proof of the Dybvig-Ingersoll-Ross-Theorem}.
\newblock \emph{Mathematical Finance}, 12\penalty0 (4):\penalty0 447--451,
  2002.

\bibitem[Kardaras and Platen(2012)]{article_Kardaras}
C.~Kardaras and E.~Platen.
\newblock {On the {D}ybvig-{I}ngersoll-{R}oss {T}heorem}.
\newblock \emph{Mathematical Finance}, 22\penalty0 (4):\penalty0 729--740,
  2012.

\bibitem[McCulloch(2000)]{article_McCulloch}
J.H. McCulloch.
\newblock {Long Forward and Zero-Coupon Rates Indeed Can Never Fall}.
\newblock \emph{Ohio State University, Working Paper}, 00\penalty0 (12), 2000.

\bibitem[Musiela and Rutkowski(1997)]{article_MusielaRutkowski1997}
M.~Musiela and M.~Rutkowski.
\newblock {Continuous-time Term Structure Models: Forward Measure Approach}.
\newblock \emph{Finance and Stochastics}, 1\penalty0 (4):\penalty0 261--291,
  1997.
\newblock URL \url{http://ideas.repec.org/a/spr/finsto/v1y1997i4p261-291.html}.

\bibitem[PLC(2014{\natexlab{a}})]{Coco_Termsheet}
Barclays PLC.
\newblock {Term-Sheet of 7.00\% Fixed Rate Resetting Perpetual Subordinated
  Contingent Convertible Securities}.
\newblock \emph{Barclay's Term-Sheet}, 2014{\natexlab{a}}.
\newblock {Available at Barclays: http://www.barclays.com/, Accessed:
  2016-07-11}.

\bibitem[PLC(2014{\natexlab{b}})]{cocoTermSheet2}
Barclays PLC.
\newblock {Term-Sheet of 8.0\% Fixed Rate Resetting Perpetual Subordinated
  Contingent Convertible Securities}.
\newblock \emph{Barclay's Term-Sheet}, 2014{\natexlab{b}}.
\newblock {Available at Barclays:
  https://www.home.barclays/content/dam/barclayspublic/docs/Investor
  Relations/esma/capital-securities-documentation/tier-1-securities/contingent-tier-1/8-fixed-rate-resetting-perpetual-subordinated-contingent-convertible-securities-731KB.pdf,
  Accessed: 2018-01-18}.

\bibitem[Protter(2005)]{Book_Protter}
P.~Protter.
\newblock \emph{{Stochastic Integration and Differential Equations}}.
\newblock Springer, 2nd edition, 2005.

\bibitem[Qin and Linetsky(2018)]{Qin-Linetsky}
L.~Qin and V.~Linetsky.
\newblock {Long term risk: A martingale approach}.
\newblock \emph{Econometrica}, 85\penalty0 (1), 2018.

\bibitem[Rogers(1997)]{Rogers}
L.C.G. Rogers.
\newblock {The Potential Approach to the Term Structure of Interest Rates and
  Foreign Exchange Rates}.
\newblock \emph{Mathematical Finance}, 7\penalty0 (2):\penalty0 157--176, 1997.

\bibitem[Rutkowski(1997)]{Rutkowski1997}
M.~Rutkowski.
\newblock {A Note on the Flesaker-Hughston Model of the Term Structure of
  Interest Rates}.
\newblock \emph{Applied Mathematical Finance}, 4\penalty0 (3):\penalty0
  151--163, 1997.

\bibitem[Schulze(2008)]{article_Schulze}
K.~Schulze.
\newblock {Asymptotic Maturity Behavior of the Term Structure}.
\newblock \emph{Bonn Econ Discussion Papers, University of Bonn}, 2008.
\newblock {Available at SSRN: http://ssrn.com/abstract=1102367, Accessed:
  2015-07-11}.

\bibitem[Shiller(1979)]{Shiller79}
R.~J. Shiller.
\newblock {The Volatility of Long-Term Interest Rates and Expectations Models
  of the Term Structure}.
\newblock \emph{The Journal of Political Economy}, 87\penalty0 (6):\penalty0
  1190--1219, 1979.

\bibitem[Yao(2000)]{Yao2000}
Y.~Yao.
\newblock {Term Structure Models: A Perspective from the Long Rate}.
\newblock \emph{North American Actuarial Journal}, 3\penalty0 (3):\penalty0
  122--138, 2000.

\end{thebibliography}
\bibliographystyle{plainnat}

\end{document}